\RequirePackage{fix-cm}

\documentclass[numbook,envcountsect,envcountsame,envcountreset,runningheads,smallextended]{svjour3}

\smartqed

\usepackage[numbers]{natbib}

\usepackage[
    colorlinks=true,
    breaklinks=true,
    bookmarks=true,
    urlcolor=blue,
    citecolor=blue,
    linkcolor=blue,
    bookmarksopen=false,
    draft=false
]{hyperref}

\usepackage{amssymb}

\usepackage{setspace}
\usepackage{graphicx}
\usepackage[dvipsnames]{xcolor}
\usepackage{amsmath,  amssymb}
\usepackage{tikz}
\usepackage{todonotes}
\usepackage{setspace}
\usepackage[a4paper, total={5.7in, 9in}]{geometry}

\usepackage{bbm}
\usepackage{nth}
\usepackage{dsfont,amsfonts}
\usepackage{subfig}
\usepackage{algorithm2e} 
\usepackage{aliascnt}

\usepackage{booktabs}
\usepackage{enumerate, enumitem}

\newcommand{\N}{{\mathbb{N}}}
\newcommand{\E}{{\mathbb{E}}}
\newcommand{\Q}{{\mathbb{Q}}}
\renewcommand{\P}{{\mathbb{P}}}
\newcommand{\q}{{\mathcal{Q}}}

\newcommand{\R}{{\mathds{R}}}

\newcommand{\mF}{{\mathcal{F}}}

\newcommand{\mB}{{\mathcal{B}}}

\newcommand{\tT}{ }

 \newcommand{\lp}[1]{L^\infty}

\newcommand{\X}{\mathcal{X}}

\DeclareMathOperator*{\limL}{{\xrightarrow{L^\infty}}}

\DeclareMathOperator*{\limas}{{\xrightarrow{a.s.}}}

\DeclareMathOperator*{\esssup}{\sup}

\DeclareMathOperator*{\Limsup}{\operatorname{LimSup}}
\DeclareMathOperator*{\Liminf}{\operatorname{LimInf}}

\DeclareMathOperator*{\clconv}{\operatorname{clconv}}

\newcommand{\Uset}{consolidated uncertainty set}

\renewcommand{\u}{{{u}}}

\newcommand{\U}{{{U}}}

\newcommand{\epi}{{ \text{epi}_U }}
\newcommand{\ba}{{ ba_{1,+} }}

\NewDocumentCommand{\LL}{O{t}}{%
   \mathcal{L}_{#1:T}%
}

\NewDocumentCommand{\bil}{mO{t}}{%
   \langle #1 \rangle_{#2,T} %
}

\usepackage{zref-clever}
\zcRefTypeSetup{theorem}{
  Name-sg = Theorem,
  name-sg = Theorem,
  Name-pl = Theorems,
  name-pl = Theorems
}

\zcRefTypeSetup{lemma}{
  Name-sg = Lemma,
  name-sg = Lemma,
  Name-pl = Lemmas,
  name-pl = Lemmas
}

\zcRefTypeSetup{proposition}{
  Name-sg = Proposition,
  name-sg = Proposition,
  Name-pl = Propositions,
  name-pl = Propositions
}

\zcRefTypeSetup{corollary}{
  Name-sg = Corollary,
  name-sg = Corollary,
  Name-pl = Corollaries,
  name-pl = Corollaries
}

\zcRefTypeSetup{definition}{
  Name-sg = Definition,
  name-sg = Definition,
  Name-pl = Definitions,
  name-pl = Definitions
}

\zcRefTypeSetup{remark}{
  Name-sg = Remark,
  name-sg = Remark,
  Name-pl = Remarks,
  name-pl = Remarks
}

\zcRefTypeSetup{example}{
  Name-sg = Example,
  name-sg = Example,
  Name-pl = Examples,
  name-pl = Examples
}

\zcRefTypeSetup{property}{
  Name-sg = Properties,
  name-sg = Properties,
  Name-pl = Properties,
  name-pl = Properties
}

\zcRefTypeSetup{item}{
  Name-sg = Item,
  name-sg = Item,
  Name-pl = Items,
  name-pl = Items
}

\newcommand{\zcenvtype}{}

\newcommand{\zcsetenvtype}[1]{%
  \gdef\zcenvtype{#1}%
  \zcsetup{reftype=#1}%
}

\newcommand{\zcclearenvtype}{%
  \gdef\zcenvtype{}%
  \zcsetup{reftype=}%
}

\newcommand{\zcrestoreenvtype}{%
  \ifx\zcenvtype\empty
    \zcsetup{reftype=}%
  \else
    \edef\zctmp{\noexpand\zcsetup{reftype=\zcenvtype}}%
    \zctmp
  \fi
}

\AddToHook{env/theorem/begin}{\zcsetenvtype{theorem}}
\AddToHook{env/theorem/end}{\zcclearenvtype}

\AddToHook{env/lemma/begin}{\zcsetenvtype{lemma}}
\AddToHook{env/lemma/end}{\zcclearenvtype}

\AddToHook{env/proposition/begin}{\zcsetenvtype{proposition}}
\AddToHook{env/proposition/end}{\zcclearenvtype}

\AddToHook{env/corollary/begin}{\zcsetenvtype{corollary}}
\AddToHook{env/corollary/end}{\zcclearenvtype}

\AddToHook{env/definition/begin}{\zcsetenvtype{definition}}
\AddToHook{env/definition/end}{\zcclearenvtype}

\AddToHook{env/remark/begin}{\zcsetenvtype{remark}}
\AddToHook{env/remark/end}{\zcclearenvtype}

\AddToHook{env/example/begin}{\zcsetenvtype{example}}
\AddToHook{env/example/end}{\zcclearenvtype}

\AddToHook{env/property/begin}{\zcsetenvtype{property}}
\AddToHook{env/property/end}{\zcclearenvtype}

\AddToHook{env/enumerate/begin}{\zcsetup{reftype=item}}
\AddToHook{env/enumerate/end}{\zcrestoreenvtype}

\begin{document}

\title{Dual Representation of Robust Risk Measures and Uncertainty Sets}

\titlerunning{Dual Representation of Robust Risk Measures}

\author{
Marlon R. Moresco
\and
Marcelo Righi
\and
Silvana M. Pesenti
}

\authorrunning{Moresco, Righi, and Pesenti}

\institute{
Marlon R. Moresco \at
Federal University of Rio Grande do Sul (UFRGS), Porto Alegre, RS, Brazil \\
\email{marlon.moresco@ufrgs.br}
\and
Marcelo Righi \at
Federal University of Rio Grande do Sul (UFRGS), Porto Alegre, RS, Brazil \\
\email{marcelo.righi@ufrgs.br}
\and
Silvana M. Pesenti \at
Department of Statistical Sciences, University of Toronto, Toronto, ON, Canada \\
\email{silvana.pesenti@utoronto.ca}
}

\date{}

\maketitle

\begin{abstract}
We consider robust risk measures that arise as worst-case values of convex risk measures evaluated on uncertainty sets. We characterize continuity properties of robust risk measures through their consolidated uncertainty sets, derive dual representations for robust risk measures, and develop a set-valued dual representation for consolidated uncertainty sets. The two dual frameworks rely on distinct geometric assumptions and are therefore complementary rather than interchangeable.

\keywords{Risk measures \and Uncertainty and ambiguity \and Convex duality \and Set-valued conjugates \and Robust finance}

\subclass{91G70 \and 91G05 \and 46N10 \and 49N15}

\noindent\textit{JEL Classification:} C61, G10, G32
\end{abstract}

\section{Introduction.}
Robust risk measures are paramount for decision making and optimisation with applications ranging from finance and insurance, economics, and risk-averse stochastic control. Distributional robust risk measures, which for simplicity we just call robust risk measures, are almost always defined as worst-case risk measures. Worst-case risk measures are the largest value a risk measure can attain when evaluated on a set of alternative random variables, where the latter is often termed the uncertainty set. Most works in the literature focus on specific choices of uncertainty sets, such as those characterised by moment constraints \cite{Cornilly2018IME,Bernard2024MF,cai2025OR}, optimal transport distances and divergences \cite{blanchet2019MOR,Pesenti2024ORL,tam2026WP}, $f$-divergences \cite{lam2016MOR,miao2025EJOR}, and distributional constraints \cite{bernard2017FS}; to name a few. Here we consider a generic uncertainty set defined as a set-valued map $u\colon \lp{} \to 2^{L^\infty}$, similar to \cite{moresco2023uncertainty}, where $u(X)$ characterises the uncertainty inherent in the risky position $X \in \lp{}$. For an uncertainty set $u$ and a risk measure $\rho\colon \lp{} \to \R$, we consider the robust risk measure
\begin{equation*}
    R^{u, \rho}(X)=\sup_{Y\in u(X)} \rho(Y)\,.
\end{equation*}
Key to the exposition is the so-called \Uset~$U$, that is the largest uncertainty set that leads to the same robust risk measure, which has been introduced in \cite{moresco2023uncertainty}. The \Uset~is instrumental for uniquely characterising properties of the robust risk measure and their dual representations.

We start by characterising continuity properties of the robust risk measure
\(R^{u,\rho}\). Since \(R^{u,\rho}\) depends on both the \Uset~\(U\) and the
risk measure \(\rho\), its continuity properties are inherited from the
interaction between these two components. As a first key result, we provide a
characterisation linking continuities of the set-valued map \(U\) with
continuities of the robust risk measure \(R^{u,\rho}\). From this part on, the manuscript is devoted to duality. We
develop separately a dual representation for the robust risk measure \(R^{u,\rho}\) and
a dual representation for the consolidated uncertainty set \(U\). The choice of
having two dual representations, one for the robust risk measure \(R^{u,\rho}\) and one
for the consolidated uncertainty set \(U\), is not merely expository, but is
dictated by a fundamental mathematical incompatibility between the two
frameworks.

Since \(R^{u,\rho}\) is a scalar-valued map on \(L^\infty\), its dual representation
naturally belongs to the classical theory of convex risk measures. It is
derived through the usual penalty-function machinery, starting from the dual
representation of \(\rho\) and incorporating the auxiliary worst-case
expectation maps
\[
g_Q(X):=\sup_{Z\in u(X)}\mathbb E_Q[-Z].
\]
On the other hand, \(U\) is a set-valued map. Hence its dual representation
must be formulated within the set-valued conjugation framework, in the sense of
Hamel, via support functions of the graph of \(U\) and an appropriate
biconjugation argument. To the authors' knowledge, the dual representation of
uncertainty sets has not been studied in the literature. These two dual representations are in general not compatible. The dual representation of
\(U\) hinges on set-concavity, whereas the dual representation of \(R^{u,\rho}\), rooted
in the standard theory of robust risk measures, relies on the usual convexity
of risk measures, which corresponds to set-convexity of \(U\). These
requirements point in opposite directions: the concavity needed for the
set-valued duality of \(U\) is not compatible, in general, with the convexity
required for the scalar duality of \(R^{u,\rho}\). Therefore, the duality of \(U\), in the
sense of Hamel, cannot be directly transferred to the duality of \(R^{u,\rho}\), in the
sense of risk measures. The two dual representations are thus complementary,
rather than duplicative, and we relate them through the corresponding support
function and penalty constructions.

From a financial perspective, this separation is natural. Convexity of the
risk measure \(R^{u,\rho}\) is well justified by diversification, meaning that combining
positions should not increase risk more than proportionally. By contrast, no
analogous principle holds for uncertainty sets. Uncertainty may increase
under aggregation even when each component is individually well specified. For
example, when the marginal distributions of two positions are exactly specified and the corresponding individual uncertainty sets are singletons, there might 
still be uncertainty about their dependence structure. In that case, the
uncertainty set of the combined position is no longer a singleton, even though
the uncertainty sets of the marginals are. Thus, the geometric behaviour that
is financially appropriate for risk measures does not naturally extend to
uncertainty sets themselves.

The distinction between the two dual representations is particularly important
because the relevant geometric assumptions differ. Convexity of risk measures
is well established as a form of diversification, whereas set-orderings of
uncertainty sets are less understood. Indeed, set-concavity is often more
intuitive than set-convexity. Consider the uncertainty set of a convex
combination of risky positions, namely
\(u(\lambda X+(1-\lambda)Y)\), where \(\lambda\in[0,1]\) and
\(X,Y\in\lp{}\). This uncertainty set includes distributional uncertainty of
each position \(X\) and \(Y\) individually, as well as uncertainty about the
dependence between \(X\) and \(Y\). Set-concavity means that the uncertainty
set of the convex combination contains the convex combination of the
uncertainty sets, that is,
\[
u(\lambda X+(1-\lambda)Y)
\supset
\lambda u(X)+(1-\lambda)u(Y),
\]
which accounts for dependence uncertainty.

Apart from the literature on robust risk measures, our work is closely related to the dual representation of risk measures and to set-valued duality. Dual representation of risk measures originated in \cite{artzner1999coherent} for coherent risk measures. It has been generalised to convex risk measures \cite{frittelli2002putting}, to risk orders \cite{drapeau2013MOR}, to real valued systemic risk measures (which are defined on the space of random vectors or sets) \cite{ararat2020dual,righi2024set-rm} and to conditional maps \cite{frittelli2011SIAM,mastrogiacomo2019DEF}. Close to our work is \cite{Righi2024} who studies examples of dual representation of robust convex risk measures and \cite{centrone2026WP-quasi-convex} who studies dual representation of robust quasi-convex risk measures. None of them, however, considers continuity of robust risk measures nor the dual representation of the uncertainty set. 

To analyse the dual representation of the uncertainty set, we rely on the theory of set-valued optimization and set-valued conjugates. The foundation for set-valued convex analysis and Fenchel conjugation was established by \cite{hamel2009duality}, who utilized complete lattices of sets to formulate a set-valued biconjugation theorem in general locally convex topological linear spaces. Comprehensive overviews and foundational treatises of this set optimization framework are provided in \cite{hamel2015set} and \cite{khan2016set}. Furthermore, the field of set-valued optimization has experienced rapid growth, with recent surveys covering advanced topics ranging from vectorisation schemes \cite{eichfelder2025two} to stochastic formulations for robust learning \cite{giovannelli2026stochastic}. Within financial mathematics, this optimization theory has been primarily applied to set-valued risk measures to handle multiple eligible assets, yielding a complete duality theory for set-valued convex measures of risk \cite{hamel2010duality}. It is worth noting that much of this applied literature focuses on mapping multivariate random variables into subsets of a finite-dimensional space (the space of eligible portfolios). This finite-dimensional approach includes the evaluation of systemic risk \cite{ararat2020dual}, pricing in conical market models with transaction costs \cite{hamel2011setvalued}, set-valued loss-based measures \cite{sun2018setvalued}, set-valued star-shaped risk measures \cite{nie2026setvalued}, and the analysis of stochastic orderings for set-valued risk measures \cite{mastrogiacomo2026stochastic}. Statistical extensions like cone distribution functions \cite{hamel2018cone} and conditional functionals of random sets \cite{fissler2026setvalued} also operate specifically within Euclidean spaces. In contrast to these finite-dimensional applications, we build upon the general infinite-dimensional conjugation frameworks, such as those for T-translative functions \cite{hamel2021setvalued} and set-valued convex compositions \cite{ararat2026setvalued}, to formulate the dual representation of set-valued uncertainty sets mapping directly into the power set of $L^\infty$.

The paper is structured as follows. \zcref[S]{framework} introduces uncertainty sets, robust risk measures, as well as the \Uset s, and recaps the relevant results needed for the exposition from \cite{moresco2023uncertainty}.  
In \zcref[S]{sec:continuity} we establish continuity of the uncertainty sets and robust risk measures, and characterise under what conditions on the uncertainty set, the robust risk measure preserves continuity. \zcref[S]{sec:dual-risk} features one of the main results of this paper, the dual representation of the robust risk measure. The second key result is stated in \zcref[S]{sec:dual-U}, which pertains to the dual of the set-valued uncertainty sets. Finally in \zcref[S]{sec:connection} we connect the two dual frameworks through their penalty functions. 

\section{Robust risk measures.}
\label{framework}
\subsection{Preliminaries and Notation.}
We consider a probability space $(\Omega,\mF,\P)$ and denote by $L^\infty := L^\infty(\Omega,\mF,\P)$ the space of $\mF$-measurable and essentially bounded random variables endowed with the supremum norm 
$    \Vert X \Vert 
    :=  \inf\big\{ m \in \R ~:~  \P( |X|>m )=0  \big\}\,.
$
Unless otherwise stated, all equalities and inequalities  are  in a $\P$-almost sure (a.s.) sense. We say that a sequence $\{X^n\}_{n\in\mathbb{N}} \subset L^\infty$ converges to $X \in L^\infty$ in $L^\infty$ if $\|X^n - X\| \to 0$ as $n\to\infty$, which we denote by $X^n  \limL  X \in \lp{}$. We further say that $\{X_n\}_{n\in\mathbb{N}} \in L^\infty$ converges to $X$ almost surely (a.s.) if $\P(\{\omega \in \Omega : X_n(\omega)\to X(\omega)\}) = 1$, and denote it by $X^n  \limas X $.

Central to the exposition are set-valued maps. To clarify the notation, we recall the addition of sets 
\begin{align*}
    A + B  := \left\{ X + Y \in \lp{0,T} : X \in A, \, Y \in B  \right\}
    \,,
    \quad \;\text{where}\quad A,B \subseteq \lp{0,T}\,.
   \end{align*}
By abuse of notation, we may denote sets consisting of a singleton by its element, i.e., $Z:= \{Z\}\subset \lp{0,T}$. We further recall the multiplication of a set $A \subseteq \lp{t}$ with a scalar $\lambda \in \R$ 
\begin{equation*}
     \lambda \, A := \big\{ \lambda \, X \in \lp{t}  ~:~X \in A  \big\}\,,
\end{equation*}
and denote the complement of a set $A \subseteq \lp{t}$ by $A^\complement := \{X \in \lp{t} ~:~ X \notin A   \}$. 
We use the following canonical order extension on the power set of $\lp{t}$  with $\tT$
\begin{equation*}
    A \preceq B \qquad \text{ if and only if } \qquad B \subseteq A + L^\infty_+,
\end{equation*}
where $L^\infty_+$  is the set of non-negative random variables. A set $A \subset \lp{t} $ is bounded  if \[\inf \{ c \in \R : \P( c \geq |X|) =1, \forall \; X \in A   \}  < \infty ,\] and bounded from below if \[\inf \{ c \in \R : \P( c \geq -X) =1, \forall \; X \in A   \}  < \infty .\]

\subsection{Risk Measures and Uncertainty sets}

We follow the concepts and terminologies introduced in \cite{moresco2023uncertainty}. However, here we work with Profit and Losses (P\&L) rather than losses as in \cite{moresco2023uncertainty}, hence, $X>0$ represents a gain and $X<0 $ a loss.
    \begin{definition}[Risk measure]\label{def-cond-rm}

    A  risk measure is a mapping $\rho : L^\infty \to \R $ and its acceptance set is defined as $A^\rho: = \{ X\in L^\infty : \, \rho(X) \leq 0\}.$
A risk measure $\rho $ may satisfy the following properties: 
\begin{enumerate}[label = $\arabic*)$]

        \item \textbf{Monotonicity:} $\rho  (X ) \leq \rho  (Y )$, for all $X , Y  \in \lp{t+1,s}$ with $Y  \leq X $.
        
        \item \textbf{Translation Invariance:} 
        $\rho \left(X+c  \right)  =\rho ( X ) - c  $, for all 
        $X  \in L^\infty  $ and all $c \in \R.$    
           \item \textbf{Convexity:}
        $\rho (\lambda \, X  + (1-\lambda )\, Y  ) \leq \lambda \, \rho  (X ) + (1-\lambda) \, \rho (Y )$, for all  $X ,Y  \in \lp{t+1,s}$ and  all $\lambda \in \R$ with $0 \leq \lambda \leq 1$.

        \item \textbf{Positive Homogeneity:}
        $\rho ( \lambda \, X ) = \lambda \,\rho (X )$, for all  $X  \in\lp{t+1,s} \,$, and all $\lambda \in \R$ with $\lambda \ge 0$.

    \end{enumerate}
We use the following commonly used terminologies. We call a risk measure a \emph{monetary risk measure}, if it satisfies monotonicity and translation invariance. A risk measure is a \emph{convex risk measure}, if it is a monetary risk measure that satisfies convexity, and it is a \emph{coherent risk measure}, if it is a convex risk measure that satisfies positive homogeneity. 
\end{definition}

Robust risk measures are typically defined through worst-case risk measures, which are the largest values a risk measure can attain when evaluated on a set of alternative random variables, where the latter is typically termed an uncertainty set. Thus we next define uncertainty sets and their desirable properties. Mathematically, an uncertainty set $u$ is a set function mapping a random variable $X$ to a subset of measurable random variables $u(X)$, which captures the uncertainty around $X$. To avoid trivial cases, we assume that the uncertainty set is non-empty and bounded from below, \cite{moresco2023uncertainty} refers to such uncertainty sets as proper uncertainty sets.

\begin{definition}[Uncertainty set]
An uncertainty set $u$  is a set function $u : L^\infty \rightarrow 2^{\lp{t}}$, such that $u(X)$ is  non-empty and bounded from below for all $X \in \lp{}$.
An uncertainty set $u$ may satisfy the following properties: 
\begin{enumerate}[label = $\arabic*)$]

    \item \textbf{Order preservation:}\label{property:order}
    Let $Y  \leq X $ with $X , Y  \in\lp{t,T}$. Then
    for each $Z \in u (X ) $ there exists a $W \in u (Y )$ such that $W\leq Z$.

    \item \textbf{Monotonicity:}\label{property:mon}   $Y  \leq X $ implies that $u  (X ) \subseteq u (Y )$, for all $X , Y  \in \lp{t,T}$.
      
    \item \textbf{Translation invariance:}\label{property:trans}
    $ u (X +c ) =  u (X )  + c $ for all $X  \in \lp{t,T}$ and $ c  \in \R$.
           
    \item \textbf{Positive homogeneity:}\label{property:pos-hom}
    $u  (\lambda \, X  ) = \lambda \,u  ( X  )  $ for all $\lambda \in\R$ with $ \lambda \ge 0$ and $X  \in\lp{t,T}$.
    \item 
   \textbf{Order convexity:} For all $Z \in u  (\lambda X  +(1-\lambda) Y ) $ there exists $X' \in u  (X ) $ and $Y' \in u  ( Y )$  such that $Z \geq \lambda X' +(1-\lambda) Y'$, for all $\lambda \in \R$ with $ 0\leq \lambda\leq 1$ and for all $X ,\, Y  \in \lp{t,T}$.

    \item \textbf{Set-convexity:}  $u  (\lambda X  + (1-\lambda)Y ) \subseteq \lambda u  (X ) + (1-\lambda) u (Y )$ for all $\lambda \in \R$ with $ 0\leq \lambda\leq 1$ and for all $X ,\, Y  \in \lp{t,T}$.

  \item \textbf{Set-concavity:}  $ \lambda u  (X ) + (1-\lambda) u (Y ) \subseteq u  (\lambda X  + (1-\lambda)Y )$ for all $\lambda \in \R$ with $ 0\leq \lambda\leq 1$ and for all $X ,\, Y  \in \lp{t,T}$.

\end{enumerate}
We call an uncertainty set a \emph{convex uncertainty set} if it satisfies order preservation, translation invariance, and order convexity. 
\end{definition}

\subsection{Robust risk measures}

With the definition of uncertainty sets, we introduce robust risk measures as the largest (worst-case) value a risk measure can attain when evaluated at random variables in the uncertainty set.

\begin{definition}[Robust risk measure]\label{defin R}
For a convex risk measure $\rho$ and an uncertainty set $u$, we define the \emph{robust risk measure} as a mapping 
$R^{\u,\rho} \colon  L^\infty  \to \R$ 
 given by
\begin{align*}
     R ^{\u,\rho}\left(X \right)
     :=
      \esssup \, \Big\{  \rho  (Y) \in \R  ~: ~Y \in u  \big(X \big)  \Big\}\,.
\end{align*}
Whenever the uncertainty set and the  risk measure are clear from the context, we simply write $R(\cdot):=R^{\u, \rho}(\cdot)$.
\end{definition}

\begin{remark}\label{remark loss to PL}
We emphasize that, in our setting, monotonicity and order preservation of an uncertainty set operate in opposite directions compared to \cite{moresco2023uncertainty}. This difference arises because, in the present study, \(X \leq Y\) means that \(Y\) is preferred over \(X\), whereas in the earlier work the variables represented losses, so \(X \leq Y\) indicated that \(X\) was preferred to \(Y\). For more details on that ordering of sets, see \cite{hamel2009duality} and references therein.

In the theory of risk measures, working with losses or with P\&L is generally interchangeable: if $\rho$ is defined on losses, the transformation $\rho^\dagger(X) := \rho(-X)$ can be interpreted as a risk measure on P\&L, with most properties requiring only minor adjustments. For robust risk measures, however, this interchangeability holds only if, in addition to $\rho^\dagger(X) := \rho(-X)$, the relation $u^\dagger(X) := -u(-X)$ also holds. In that case, it follows that $R^\dagger(X) = R(-X)$. While the interpretation of $X$ versus $-X$ is straightforward, the meaning of $-u(-X)$ is generally less clear. Furthermore, while $\rho^\dagger$ preserves the convexity of $\rho$, if $u$ is set-convex $u^\dagger$ is set-concave. 
\end{remark}

Next we recall results from \cite{moresco2023uncertainty} which are needed for the exposition. In particular, we are interested which properties a robust risk measure inherits from its uncertainty set. \zcref[S]{prop:u_to_R} provides sufficient conditions for this inheritance, in particular, it demonstrates that a convex uncertainty set guarantees a convex robust risk measure. For further properties and equivalences between properties of $R$ and $u$, see \cite{moresco2023uncertainty}.

\begin{proposition}[Proposition 2 of \cite{moresco2023uncertainty}]\label{prop:u_to_R} 
Let $R$ be a  robust risk measure with uncertainty set $u$. 
Then the following hold:
\begin{enumerate}[label = $\arabic*)$]
   
    \item\label{u monotone - R monotone} 
    If $u$ is monotone or order preserving, then  $R$ is monotone.
        
    \item\label{u trans -> R trans} 
    If $u$ is translation invariant, then $R$ is translation invariant.
    
    \item  \label{u ph - R ph} 
    If $\rho$ and $u$ are positive homogeneous, then $R$ is  positive homogeneous.

    \item \label{u conv - R conv} If $u$ is set-convex or order convex, then $R$ is convex.

    \item \label{u conc - R conc} Let $\rho$ be concave\footnote{For this property only, we allow  $\rho$ to not be convex.}. If $u$ is set-concave,
    then $R$ is concave.    
    
\end{enumerate}
\end{proposition}

Given the non-uniqueness of uncertainty sets, specifically, that uncountably many uncertainty sets give rise to the same robust risk measure, we define the so-called \Uset s. The \Uset\, associated with a robust risk measure is the largest uncertainty set that gives rise to the same robust risk measure. Below we define the \Uset\ using its equivalent representation, that is the set of all random variables that have a risk less than that of the robust risk measure; the equivalence of the two representations is given in \zcref[S]{lemma u to U} below. We refer to Section 3.2 in \cite{moresco2023uncertainty} for further discussion and interpretation of the \Uset.

\begin{definition}[Consolidated uncertainty set]\label{defin U set}
Let $R$ be a  robust risk measure with uncertainty set $\u$ and  risk measure $\rho $. Its \emph{\Uset} is the uncertainty set $U $, defined for all  $X \in \lp{t,T}$ by
\begin{align*}
       U (X ) := \left\{ Y \in \lp{t}~:~ \rho (Y) \leq R^{\u,\rho}(X )  \right\}.
\end{align*} 
\end{definition}

We call $U$ a \emph{convex \Uset}, if it satisfies monotonicity, translation invariance, and set-convexity. We call $U$ a \emph{concave \Uset}, if it satisfies monotonicity, translation invariance, and set-concavity.

The next lemma collects different representations of \Uset s.

\begin{lemma}[Lemma 3 of \cite{moresco2023uncertainty}]\label{lemma u to U}
Let $R$ be a robust risk measure with uncertainty set $u$ and \Uset \; $U$. Then it holds
for all $ X \in \lp{t+1,T}$ that
\begin{enumerate}[label = $\arabic*)$]

    \item \label{ item equality bet unc sets} 
    $  U (X) =\bigcup \Big\{ u'(X) \subseteq \lp{t} :   
   R^{u'}(X) = R^{u}(X)     \Big\} $.

    \item \label{equality of the risk measure for two unc sets} 
    $R^U(X) = R^u (X)$.
   
    \item 
    \label{lemma u to U: 4}
    $ U (X) = A^\rho -  R^U (X)$.
\end{enumerate}

\end{lemma}
We highlight that a \Uset \, is a convex \Uset\, if and only if it is also a convex uncertainty set (Corollary 1 of \cite{moresco2023uncertainty}). We postpone the discussion of concave \Uset s to \zcref[S]{sec:dual-U}, which is needed for the dual representation of a \Uset.

Finally, we recall the one-to-one correspondence between the properties of a robust risk measure and its \Uset.

\begin{theorem}
[Theorem 2 of \cite{moresco2023uncertainty}]
\label{theo:equiv_R_U}
Let $R$ be a robust risk measure with  uncertainty set $u$ and denote by $U$ its associated \Uset. 
Then, the following hold:
\begin{enumerate}[label = $\arabic*)$]
    
    \item
    $R$ is monotone if and only if  $U$ is monotone. \label{theo:equiv_R_U_monotone}
    
    \item
    $R$ is translation invariant if and only if $U$ is translation invariant. \label{theo:equiv_R_U_trans_inv}

    \item
    \label{theo:equiv_R_U_ph} 
    Let $\rho$ be positive homogeneous. Then $R$ is positive homogeneous if and only if $U$ is positive homogeneous.

    \item
  $R$ is convex if and only if $U$ is set-convex.
        \label{theo:equiv_R_U_conv}
       
   \item
$R$ is concave if and only if $U$ is set-concave.
   \label{theo:equiv_R_U_conc}
        
\end{enumerate}
\end{theorem}

\begin{proof}\zcref[S]{theo:equiv_R_U_monotone,theo:equiv_R_U_trans_inv,theo:equiv_R_U_ph,theo:equiv_R_U_conv} follows directly from Theorem 2 of \cite{moresco2023uncertainty}. \zcref[S]{theo:equiv_R_U_conc} is a slightly generalization of their result, here we provide a proof.  
Assume first that $R$ is concave. Let $X_1,X_2\in L^\infty$, $\lambda\in[0,1]$, and take $Y_1\in U(X_1)$ and $Y_2\in U(X_2)$. Then $\rho(Y_1)\le R(X_1)$ and $\rho(Y_2)\le R(X_2)$. By convexity of $\rho$, and concavity of $R$
\[
\rho(\lambda Y_1+(1-\lambda)Y_2)
\le
\lambda\rho(Y_1)+(1-\lambda)\rho(Y_2) \le
\lambda R(X_1)+(1-\lambda)R(X_2) \leq R(\lambda X_1+(1-\lambda)X_2).
\]
Therefore,
$\rho(\lambda Y_1+(1-\lambda)Y_2) \le R(\lambda X_1+(1-\lambda)X_2), $
which implies $\lambda Y_1+(1-\lambda)Y_2\in U(\lambda X_1+(1-\lambda)X_2)$. Thus,
\[
\lambda U(X_1)+(1-\lambda)U(X_2)
\subseteq
U(\lambda X_1+(1-\lambda)X_2),
\]
and $U$ is set-concave.

Conversely, assume that $U$ is set-concave. Fix $X_1,X_2\in L^\infty$ and $\lambda\in[0,1]$. Note that for any $X \in \lp{}, \rho(\rho(0) -R(X)) = R(X)$. Hence $\rho(0)-X_i\in U(X_i)$ for $i=1,2$. Since $U$ is set-concave, we get
\[
\lambda(\rho(0)-R(X_1))+(1-\lambda)(\rho(0)-R(X_2))=\rho(0)-(\lambda R(X_1)+(1-\lambda)R(X_2))
\in
U(\lambda X_1+(1-\lambda)X_2).
\]
Therefore,
\[
R(\lambda X_1+(1-\lambda)X_2) \geq \rho( \rho(0)-(\lambda(R(X_1))+(1-\lambda)(R(X_2)))) = \lambda R(X_1)+(1-\lambda)R(X_2).
\]
Thus, $R$ is concave.
\end{proof}

Next, we show that the robust risk measure is attained.

\begin{lemma}
    For any uncertainty set $\u$ and associated \Uset~$\U$, it holds that the supremum in the definition of $R$ is attained in the \Uset~$\U$. That is, there exists a $Y\in\lp{0,T}$ such that  $Y \in U(X )$, and
\begin{equation*}
    R^{\u}  (X ) 
    = \rho  (Y )\,,
\end{equation*}
\end{lemma}

 \begin{proof}
    Recall that $R^{\u}  (X ) = R^{\U}  (X )$. Set $Y  := \rho(0)-R  (X )$, then by translation invariance  of $\rho $ it holds that $\rho (Y )= \rho (\rho(0)-R  (X ) ) = \rho(0) - (\rho(0)-R  (X )) = R  (X )$. Finally $Y $ belongs to the uncertainty set since $Y  = \rho(0)-R  (X ) \in U  (X ) = \left\{ Y \in \lp{t+1} : \rho  (Y) \leq R  (X )  \right\} $.
 \end{proof}

We end this section by establishing some useful properties of the \Uset.

\begin{lemma}\label{lemma U containted in U + Y}
    Let $u$ be an order preserving uncertainty set and $U$ its \Uset. Then, for all $X \in \lp{}$ we have the following: 
    \begin{enumerate}[label = $\arabic*)$]
        \item\label{lemma U containted in U + Y: item Y non negative} if $Y \in \lp{+}_+$ then  $U(X) +Y \subseteq U(X)$.

        \item\label{lemma U containted in U + Y: item Y acceptable 2} if $\rho$ is normalized and $U(X) +Y\subseteq U(X) $,  then $Y \in A^\rho $.
        
        \item\label{lemma U containted in U + Y: item Y real} if $c \in \R$ then $U(X) +c \subseteq U(X) $ if and only if $c\geq 0$. With equality only if $c=0$.
    \end{enumerate}

\end{lemma}

 \begin{proof}
\zcref[S]{lemma U containted in U + Y: item Y non negative}, take any $Y \in \lp{+}_+$. Then,  for any $Z\in U(X)$, as $Z-Y\leq Z$, monotonicity implies $\rho(Z-Y) \geq \rho(Z)$. Hence,  $ U(X) + Y = \{ Z \in \lp{}: \rho(Z-Y) \leq R(X)\} \subseteq \{ Z \in \lp{}: \rho(Z) \leq R(X)\}  = U(X)$.

\zcref[S]{lemma U containted in U + Y: item Y acceptable 2}, we need to show that if $Y \notin  A^\rho $ then there exists a $Z \in U(X)+Y$ such that $Z \notin U(X) $. Take $Z=Y-R(X)$, clearly, $Z \in U(X)+Y$. Furthermore, $\rho(Z) = R(X) + \rho(Y) > R(X) $. Hence, $Z \notin U(X)$.

\zcref[S]{lemma U containted in U + Y: item Y real}, is a direct consequence of \zcref[S]{lemma U containted in U + Y: item Y non negative,lemma U containted in U + Y: item Y acceptable 2}. 

\end{proof}

\section{Continuities of sets and risk measures.}\label{sec:continuity}

This section is devoted to continuity properties of robust risk measures, which is of independent interest but also needed for the dual representation of robust risk measures. 

\subsection{Continuity of set-valued maps}
We first recall some definitions of limits of sets.

\begin{definition}

    The distance between  a random variable $X \in \lp{T}$ and a set $A \subseteq \lp{T}$ is given by
    \begin{equation*}
         d(X, A) := \inf_{Y\in A} \Vert Y-X \Vert.
    \end{equation*}
\end{definition}

\begin{definition}[Limits of sets]
Let $\{A ^n\}_{n\in\N}$ be a sequence of sets such that $A ^n \subseteq \lp{t}$ for all $n\in\N$. We recall the following limits of set functions \citep{aubin2009set}:

\begin{enumerate}[label = $\arabic*)$]
    \item 
    \textbf{Outer limit:}\footnote{Note that $\limsup$ is the limit supremum of functions and sequences, while $\Limsup$ is the outer limit of sets.} 
    The outer limit of the sequence of sets $\{A ^n\}_{n\in\N}$ is given by
    \begin{equation*}
    \Limsup_{n \to \infty} A ^n 
    := \big\{ X \in \lp{t} ~:~ \liminf_{n \to \infty} d(X,A ^n) = 0 \big\}     \,.
    \end{equation*}

    \item 
    \textbf{Inner limit:} 
    The inner limit of the sequence of sets $\{A ^n\}_{n\in\N}$ is given by
    \begin{equation*}
         \Liminf_{n \to \infty} A ^n := \big\{ X \in \lp{t} ~:~\lim_{n \to \infty} d(X,A ^n) = 0 \big\} 
         \,.
        \end{equation*}
    \item 
    \textbf{Limit:} If $\Liminf\limits_{n \to \infty} A ^n= \Limsup\limits_{n \to \infty} A ^n$, then the limit of $\{A ^n\}_{n\in\N}$ exists and is defined as 
    \begin{equation*}
    \lim_{n \to \infty} A ^n 
    :=
    \Liminf_{n \to \infty} A ^n= \Limsup_{n \to \infty} A ^n \,.    
    \end{equation*}
\end{enumerate}
\end{definition}

Note that the above limits are defined as norm-limits. That is, the outer limit is the set of all cluster points in the norm of sequences, $\{a_n\}_{n\in\N}$ where $a_n \in A_n$, and the inner limit is the set of all limit points in the supremum norm.

\begin{definition}[Continuity of set-valued functions]
   A set function $u
   : \lp{t,T }\rightarrow  2^{\lp{t}} $ may satisfy the following continuity properties:
   \begin{enumerate}[label = $\arabic*)$]
\item  \textbf{Lower semi-continuity (lsc):} if
for all $\{X^n \}_{n\in\N} \subseteq \lp{t+1,T}$ such that $X^n  \limL  X \in \lp{t+1}$ it holds that 
\[ u (X  ) \subseteq \Liminf\limits_{n \rightarrow \infty} u  (X ^n), \]

\item  \textbf{Upper semi-continuity (usc):} if 
for all $\{X^n \}_{n\in\N} \subseteq \lp{t+1,T}$ such that $X^n  \limL  X \in \lp{t+1}$ it holds that 
\[ \Limsup\limits_{n \rightarrow \infty} u  (X ^n) \subseteq u (X  ), \]

        \item \textbf{ Fatou lower semi-continuity (Fatou lsc):}  If for all bounded $\{X^n \}_{n\in\N} \subseteq \lp{t+1,T}$ such that $X^n  \limas X \in \lp{t+1}$  it holds that 
        \[u (X  ) \subseteq \Liminf\limits_{n \rightarrow \infty} u  (X ^n) .\] 

   \item \textbf{Fatou upper semi-continuity (Fatou usc):} If for all bounded 
$\{X^n\}_{n\in\N} \subseteq \lp{t+1,T}$ such that 
$X^n \limas X \in \lp{t+1,T}$ it holds that
\[
\Limsup\limits_{n \rightarrow \infty} u(X^n) \subseteq u(X).
\]

       \item  \textbf{Lipschitz continuity:}  if for all $X  , Y  \in \lp{t,T}$ it holds that 
    \begin{equation*}
        U(X ) \subseteq U(Y ) + \Vert X  - Y  \Vert \;  \mB,
    \end{equation*}  
    where $\mB := \{\lambda \in \R ~:~ | \lambda | \leq 1    \} $ is the unit ball on $\R$.

   \end{enumerate}
   \end{definition}

\begin{remark}
The definition of Lipschitz continuity above is slightly stronger than the usual one, for details, see e.g., \cite{aubin2009set}. Note that for lsc and usc, both norm-based and Fatou, the convergence of $u(X^n)$ is with respect to the strong topology. Thus, the Fatou lsc and Fatou usc conditions introduced above are stronger than their
purely almost-sure counterparts. Indeed, the set limits $\Liminf$ and $\Limsup$
used here are the Painlevé--Kuratowski limits induced by the $L^\infty$ norm.
Hence, although the sequence $X^n$ is assumed to converge almost surely to $X$,
the convergence of the corresponding selections from $u(X^n)$ is still required
in the $L^\infty$ norm. This is a strong requirement: even singleton-valued maps,
such as $u(X)=\{X\}$, need not satisfy Fatou lsc under this definition, since
bounded almost sure convergence does not imply convergence in $L^\infty$.
A weaker version could be obtained by replacing the norm-based
Painlevé--Kuratowski limits with sequential Painlevé--Kuratowski limits defined
directly through almost sure convergence. The results in this study remain
valid under this weaker formulation, provided the corresponding scalar Fatou
assumptions are imposed where needed.
\end{remark}

The next result states that monotonicity and translation invariance implies Lipschitz continuity of set function; a property that is known for risk measures. Indeed any risk measure that is monotone and translation invariant is Lipschitz continuous.

\begin{lemma}\label{lemma lipschitz}
    If a set function $U
   : \lp{t,T }\rightarrow  2^{\lp{t}} $ is monotone and translation invariant, then $\U$ is Lipschitz continuous.
\end{lemma}

 \begin{proof}
    Let $X ,Y  \in \lp{t,T}$, then we clearly have that $Y  - \Vert X  - Y  \Vert \leq X  $. Therefore 
    \begin{align*}
        U  (X ) 
        &\subseteq U \big(Y  -\Vert X  - Y  \Vert  \big) 
        = U (Y ) -\Vert X  - Y  \Vert  
        \subseteq U (Y ) - \mB \Vert X  - Y  \Vert  
        = U (Y ) + \mB \Vert X  - Y  \Vert  ,
    \end{align*}
    where the first inclusion is due to the monotonicity of $\U$, and translation invariance yields the equality. The last inclusion follows as $0 \in \mB$ and the last equality follows from the symmetry of the unit ball, i.e. $\mB = - \mB$. 
\end{proof}

\subsection{Continuities of risk measures}

We recall well-known definitions of continuity of risk measures.

\begin{definition}[Continuities of risk measures]
Let $R \colon L^\infty\to \R$ be a risk measure. Then we say that $R $ is
    \begin{enumerate}[label = $\arabic*)$]
        \item \textbf{Lower semi-continuous (lsc):} if $R (X ) \leq \liminf\limits_{n \rightarrow \infty} R (X^n )$ for all $\{X^n \}_{n\in\N} \subseteq \lp{t+1,T}$ such that $X^n  \limL  X \in \lp{t+1}$.
        
        \item \textbf{Upper semi-continuous (usc):} if $\limsup\limits_{n \rightarrow \infty} R (X^n ) \leq R (X )$ for all $\{X^n \}_{n\in\N} \subseteq \lp{t+1,T}$ such that $X^n  \limL  X \in \lp{t+1}$.
        
        \item \textbf{Continuous:} if $R$ is both lower and upper semi-continuous.
        
        \item \textbf{Fatou lower semi-continuous (Fatou lsc):} if $R (X ) \leq \liminf\limits_{n \rightarrow \infty} R (X^n )$,  for all bounded sequences $\{X^n \}_{n\in\N} \subseteq \lp{t+1,T}$ such that $X^n  \limas X \in \lp{t+1}$.

\item \textbf{Fatou upper semi-continuous (Fatou usc):} if $ \limsup\limits_{n \rightarrow \infty} R(X^n) \leq R(X),$
for all bounded sequences $\{X^n\}_{n\in\N} \subseteq \lp{t+1,T}$ such that 
$X^n \limas X \in \lp{t+1,T}$.
        
        \item \textbf{Lipschitz continuous:} if  $|R (X )- R (Y )| \leq  \Vert X  - Y  \Vert  $ for all $ X , Y  \in \lp{t+1,T} $.

        \item \textbf{Weak* lsc}:  if the lower level sets $\{X \in L^\infty : R(X) \le c\}$ are weak* closed for all $c \in \mathbb{R}$.
    \end{enumerate}
\end{definition}

With the above continuity definitions of risk measures and uncertainty sets, we next derive the equivalence between continuity of robust risk measures and continuity of their uncertainty sets.

\begin{proposition}\label{prop: continuity}
    Let $R$ be a  robust risk measure associated with an uncertainty set $\u$ and its \Uset \; $\U$. Then, the following hold:
    \begin{enumerate}[label = $\arabic*)$]
        \item \label{prop: continuity - lower}  If $u $ is lower semi-continuous, then $R $ is lower semi-continuous.
        
        \item\label{prop: continuity - upper} If $u $ is upper semi-continuous and $u  (Y )$ is compact with respect to the strong topology in $\lp{}$ for any $Y \in \lp{t+1,T}$,  then $R $ is upper semi-continuous.

        \item \label{prop: continuity - fatou}  If $u $ is Fatou lsc,  then $R $ is Fatou lsc.
        
        \item \label{prop: continuity - fatou usc} If $u$ is Fatou usc and, for every bounded sequence
$\{X^n\}_{n\in\N}\subseteq \lp{t+1,T}$ such that $X^n\limas X$, the set
$\bigcup_{n\geq N}u(X^n)$ is contained in a compact subset of $\lp{}$ for some
$N\in\N$, then $R$ is Fatou usc.
      
        \item \label{prop: continuity - lipschitz} If $\u$ is Lipschitz continuous, then $R$ is Lipschitz continuous.
    \end{enumerate}
\end{proposition}

 \begin{proof}
To ease notation, we set $\rho  ( u) := \{ \rho (Y) : Y \in u \}$ for any $u \subseteq \lp{t+1}$. Of course we have that if $u \subseteq u'$, then $\rho  (u) \subseteq \rho (u')$. Additionally, fix $X ^n$ such that it converges to $X $, as $n \to \infty$ in $\lp{}$.

 \zcref[S]{prop: continuity - lower},  by Proposition 1.2.2 of \cite{aubin2009set}, since $\rho $ is Lipschitz continuous with respect to the $\lp{}$ norm, we have $\rho  \big( \Liminf\limits_{n \rightarrow \infty} u^n \big) \subseteq \Liminf\limits_{n \rightarrow \infty} \rho (u^n)  $ for any sequence of sets $u^n$, then \begin{align*}
            R  (X ) 
            &= \esssup \left\{ \rho  \big( u \left(X  \right)\big) \right\}
            \\ &\leq \esssup \left\{ \rho  \left( \Liminf\limits_{n \rightarrow \infty} u \left( X ^n \right)\right) \right\} 
            \\ &\leq \esssup \left\{ \Liminf\limits_{n \rightarrow \infty} \rho  \big(  u ( X ^n)\big) \right\}
 \\ &\leq     \Liminf\limits_{n \rightarrow \infty} \left( \esssup  \left\{\rho  \left(  u ( X ^n)\right) \right\} \right)
 \\ &=\liminf\limits_{n \rightarrow \infty}  R\big( X ^n\big) .
    \end{align*}

 \zcref[S]{prop: continuity - upper}, by Lemma 1.1.6 of \cite{aubin2009set}, if $\rho $ is  Lipschitz continuous,  then $\limsup\limits_{n \rightarrow \infty} \left( \esssup \rho (u^n) \right) \leq \esssup \rho  \big(\Limsup\limits_{n \rightarrow \infty} u^n \big)$ for any sequence of compact sets $u^n$. Therefore,
    \begin{align*}
        \limsup\limits_{n \rightarrow \infty} R \left(X ^n\right) 
        &= \limsup\limits_{n \rightarrow \infty}  \left(\esssup \left\{ \rho \left(u  (X ^n )\right) \right\} \right)
        \\ &\leq    \esssup \left\{ \rho  \left(\Limsup\limits_{n \rightarrow \infty} u  (X ^n)  \right)\right\}
        \\ &\leq    \esssup \{ \rho  (  u  (X ) )\} 
        \\ &=  R \left(X \right)\,.  
    \end{align*}

\zcref[S]{prop: continuity - fatou} follows by using similar steps as \zcref[S]{prop: continuity - lower}, with the only difference that $X ^n$ now is bounded and converges to $X $ almost surely, as $n \to \infty$.

\zcref[S]{prop: continuity - fatou usc}, let $\{X^n\}_{n\in\N}\subseteq \lp{t+1,T}$ be bounded and such that $X^n\limas X$. If $\limsup\limits_{n\to\infty} R(X^n)=-\infty$, there is nothing to prove.
Otherwise, by the compactness assumption and the Lipschitz continuity of $\rho$, 
$\limsup\limits_{n\to\infty} R(X^n)<\infty$. Hence, there exists a subsequence
$\{X^{n_k}\}_{k\in\N}$ such that $R(X^{n_k})\to \limsup\limits_{n\to\infty} R(X^n)$.
For each $k\in\N$, choose $Z^k\in u(X^{n_k})$ such that
\[
R(X^{n_k})-\frac{1}{k}\leq \rho(Z^k)\leq R(X^{n_k}).
\]
By assumption, there exists $N\in\N$ and a compact set $K\subseteq \lp{}$ such that $\bigcup_{n\geq N}u(X^n)\subseteq K$. Hence, passing to a further subsequence if necessary, there exists $Z\in\lp{}$ such that $Z^k\limL Z$. Since $Z^k\in u(X^{n_k})$, Fatou usc of $u$ yields
\[
Z\in \Limsup_{k\to\infty}u(X^{n_k})\subseteq u(X).
\]
By Lipschitz continuity of $\rho$, $\rho(Z^k)\to\rho(Z)$. Therefore,
\[
\limsup_{n\to\infty}R(X^n)=\lim_{k\to\infty}R(X^{n_k})=\lim_{k\to\infty}\rho(Z^k)=\rho(Z)\leq R(X).
\]
Thus, $\limsup_{n\to\infty}R(X^n)\leq R(X)$.

    \zcref[S]{prop: continuity - lipschitz}, let $X , Y  \in \lp{t+1,T}$ such that $\Vert X  - Y  \Vert  
 >0$ and $\mB$ the unit ball on $\R$, note that $\mB=-\mB$. Then, by Lipschitz continuity of $\u$, it holds that
    \begin{align*}
         R (X ) 
         &= \esssup \big\{ \rho (Z) : Z \in u  (X  ) \big\}
         \\ &
         \leq \esssup \big\{  \rho (Z) : Z \in u  (Y  ) + \Vert X  - Y  \Vert \, \mB \big\}
          \\ &
         = \esssup \big\{  \rho (Z) : Z \in u  (Y  ) - \Vert X  - Y  \Vert \, \mB \big\}
         \\ &=
         \esssup \big\{  \rho (Z - b\, \Vert X  - Y  \Vert ) : Z \in u  (Y  ), \; \Vert b \Vert \leq 1 \big\}
         \\ 
         &=
         \esssup \big\{  \rho (Z) + \Vert X  - Y  \Vert    : Z \in u  (Y  ) \big\}
         \\ &
         =R (Y ) + \Vert X  - Y  \Vert  \,. 
    \end{align*}
    
    By reversing the roles of $X$ and $Y$, we have that $|R (X )- R (Y )| \leq  \Vert X  - Y  \Vert  $.
    \end{proof}

\begin{remark}
For an uncertainty set $\u$, we can always construct a compact uncertainty set $\u'$ such that $R^{\u} = R^{\u'}$. The compact uncertainty set $\u'$ can be chosen as $u' ( X ) := \{\rho(0)-R^{\u}  ( X )  \} $. Since $u'$ is  a singleton, it is always compact with respect to $\lp{}$.     
\end{remark}   

As a last result of this section, we develop the equivalence between continuity of the \Uset~and continuity of the robust risk measure.

\begin{theorem}\label{theo continuity}
Let $R$ be a  robust risk measure with  associated \Uset ~ $\U$. 
Then, the following hold:
    \begin{enumerate}[label = $\arabic*)$]
        \item \label{theo continuity: lower} $R $ is lower semi-continuous  if and only if $U$ is lower semi-continuous.

         \item\label{theo continuity: upper}  $R $ is upper semi-continuous if and only if $U$ is upper semi-continuous.

          \item \label{theo continuity: fatou}  $R $ is Fatou lsc if and only if $U$ is Fatou lsc. 
          \item \label{theo continuity: fatou usc} $R$ is Fatou usc if and only if $U$ is Fatou usc.
        \item \label{theo continuity: lipschitz}  $R$ is Lipschitz continuous if and only if $\U$ is Lipschitz continuous.
    \end{enumerate}
\end{theorem}

 \begin{proof}
 In light of \zcref[S]{prop: continuity} we only need to prove the only if direction of each item besides \zcref[S]{theo continuity: upper,theo continuity: fatou usc}. 
 
\zcref[S]{theo continuity: lower}, for $X ^n \limL X $, we show the following chain of inclusions 
\begin{subequations}\label{pf:inclusion}
     \begin{align}
          \Liminf\limits_{n \rightarrow \infty}  U  (X ^n) 
          &\supseteq 
          \big\{ Y \in \lp{t+1} :  \rho  (Y) \leq \liminf\limits_{n \rightarrow \infty} R  (X ^n) \big\} 
          \label{pf:inclusion-1}
          \\ 
          &\supseteq U (X )
          \,.
          \label{pf:inclusion-2}
     \end{align}
\end{subequations}

We start by showing \eqref{pf:inclusion-1}. 
     Take $Y \in \lp{t+1}$ such that  $ \rho  (Y) \leq \liminf\limits_{n \rightarrow \infty} R  (X ^n) $  and define the sequence 
     \begin{equation*}
     Y^n := Y - \inf \{ R  (X ^m) : m\geq n\} +  \liminf\limits_{n \rightarrow \infty}  R  (X ^n)\,,
     \quad n\in\N\,.
     \end{equation*}
     Clearly $Y^n \limL Y$ and $\rho (Y^n) = \rho (Y) + \inf \{ R  (X ^m) : m\geq n\} -  \liminf\limits_{n \rightarrow \infty}  R  (X ^n) $ converges to $\rho (Y)$ in $\lp{}$, as $n \to \infty$. Furthermore, as $\rho (Y) \leq \liminf\limits_{n \rightarrow \infty} R  (X ^n)$, we have
     \begin{align*}
         \rho (Y^n) &= \rho (Y) + \inf \{ R  (X ^m) : m\geq n\} -  \liminf\limits_{n \rightarrow \infty}  R  (X ^n) 
         \\ &\leq \liminf\limits_{n \rightarrow \infty}  R  (X ^n) + \inf \{ R  (X ^m) : m\geq n\} -  \liminf\limits_{n \rightarrow \infty}  R  (X ^n) 
         \\ &=\inf \{ R  (X ^m) : m\geq n\} 
         \\ &\leq  R  (X ^n)\,.
     \end{align*}
    Thus, we constructed a sequence $\{Y^n\}_{n\in\N}$, converging to $Y$ that satisfies $\rho  ( Y^n) \leq R  (X ^n)  $ for all $n$. Now by the definition of the inner limit and the definition of $\U$ if there exists a sequence $Y^n \limL Y$ satisfying $\rho  ( Y^n) \leq R  (X ^n)  $ for all $n$, then $Y \in \Liminf\limits_{n \rightarrow \infty} U  (X ^n) $. Hence, we obtain the set inclusion \eqref{pf:inclusion-1}. 
    
    The inclusion \eqref{pf:inclusion-2} follows by lower semi-continuity of $R$, indeed  \begin{align*}
          U (X ) &= \{ Y \in \lp{t+1} :  \rho  (Y) \leq R  (X ) \} 
          \subseteq
          \{ Y \in \lp{t+1} :  \rho  (Y) \leq \liminf\limits_{n \rightarrow \infty} R  (X ^n) \}.
         \end{align*}

\zcref[S]{theo continuity: upper}, for the only if direction,
recall that  $Y \in \Limsup\limits_{n \rightarrow \infty} U  (X ^n) $ if there exists a sequence $Y^n$ satisfying $\rho  ( Y^n) \leq R  (X ^n)  $ for all $n \in \N$, and there exists a sub-sequence $Y^{n_k} \limL Y$. As $\rho $ is Lipschitz continuous, it is also lower semi-continuous. This, together with   upper semi-continuity  of $R$ yields     
     \begin{align*}
         \rho  (Y) \leq \liminf\limits_{n \rightarrow \infty} \rho (Y^{n_k}) \leq  \limsup\limits_{{n_k} \rightarrow \infty} \rho  (Y^{n_k})  \leq  \limsup\limits_{{n_k} \rightarrow \infty} R  (X ^{n_k}) \leq  R  (X ) \, .
     \end{align*}  
     Hence, 
     \begin{align*}
          \Limsup\limits_{n \rightarrow \infty}  U  (X ^n) 
           \subseteq \{ Y \in \lp{t+1} :  \rho  (Y) \leq R  (X ) \} 
          = U (X )\,.
     \end{align*}

Conversely, assume that $U$ is upper semi-continuous and let
$X^n \limL X$. We prove that $
\limsup_{n\to\infty} R(X^n) \leq R(X).$
Suppose, by contradiction, that
$\limsup_{n\to\infty} R(X^n)>R(X)$.
First, assume that $\limsup_{n\to\infty} R(X^n)<\infty$. Take a subsequence
$\{X^{n_k}\}_{k\in\N}$ such that
$R(X^{n_k})\to \limsup_{n\to\infty}R(X^n).$
Define $Y^k:=\rho(0)-R(X^{n_k}), \:k\in\N.
$ By translation invariance of $\rho$,
$\rho(Y^k)=R(X^{n_k}), $
and therefore $Y^k\in U(X^{n_k})$ for every $k\in\N$. Moreover,
\[
Y^k \limL \rho(0)-\limsup_{n\to\infty}R(X^n).
\]
Hence, $\rho(0)-\limsup_{n\to\infty}R(X^n)
\in
\Limsup_{n\to\infty}U(X^n).$
Since $U$ is upper semi-continuous, this implies
$\rho(0)-\limsup_{n\to\infty}R(X^n)\in U(X).$
Thus,
\[
\limsup_{n\to\infty}R(X^n)
=
\rho\left(\rho(0)-\limsup_{n\to\infty}R(X^n)\right)
\leq R(X),
\]
which contradicts the assumption.

It remains to rule out the case
$\limsup_{n\to\infty}R(X^n)=+\infty$. Choose $c\in\R$ such that
$\rho(c)>R(X)$. Since $\limsup_{n\to\infty}R(X^n)=+\infty$, there exists a
subsequence $\{X^{n_k}\}_{k\in\N}$ such that
$\rho(c)\leq R(X^{n_k})$ for all $k\in\N$. Hence $c\in U(X^{n_k})$ for all
$k\in\N$. Since the constant sequence $c$ converges to $c$ in $L^\infty$, we have
$c\in \Limsup_{n\to\infty}U(X^n).$ 
By upper semi-continuity of $U$, $c\in U(X)$, which implies $\rho(c)\leq R(X), $
contradicting the choice of $c$. Therefore,
$
\limsup_{n\to\infty} R(X^n) \leq R(X),
$ and $R$ is upper semi-continuous.

\zcref[S]{theo continuity: fatou}, the only if direction follows analogously to the proof of the inclusions in \eqref{pf:inclusion}, where the assumption $X^n \xrightarrow{L^\infty} X$ is replaced by bounded sequences such that $X^n \xrightarrow{\text{a.s.}} X$. By the definition of the inner limit and of $U$, $Y \in \Liminf_{n \to \infty} U(X^n)$ if there exists a sequence $Y^n \xrightarrow{L^\infty} Y$ satisfying $\rho(Y^n) \le R(X^n)$ for all $n \in \mathbb{N}$. Since $R$ is Fatou lsc, the inequality $R(X) \le \liminf_{n \to \infty} R(X^n)$ holds. Therefore, \eqref{pf:inclusion-2} remains valid, and the sequence $Y^n$ constructed previously satisfies the required acceptance conditions.

\zcref[S]{theo continuity: fatou usc} follows analogously to \zcref[S]{theo continuity: upper}. 

 \zcref[S]{theo continuity: lipschitz}, by Lipschitz continuity of $R$,
  we  have that $R  (X ) \leq R  (Y ) + \Vert X  - Y  \Vert $.
 Next by translation invariance of $\rho $, it holds 
\begin{align*}
    U (X ) 
    &=
    \big\{ Z \in \lp{t+1} ~:~  \rho  (Z) \leq  R  (X ) \big\}
    \\ 
    &\subseteq 
    \big\{ Z \in \lp{t+1} ~:~  \rho  (Z) \leq  R  (Y ) + \Vert X  - Y  \Vert  \big\}
    \\
    &= 
    \big\{ Z \in \lp{t+1} ~:~  \rho  \big(Z  + \Vert X  - Y  \Vert \big) \leq R  (Y )   \big\}
    \\ 
    &=
    \big\{ Z -   \Vert X  - Y  \Vert  \in \lp{t+1} ~:~  \rho  (Z ) \leq   R  (Y )  \big\} 
    \\
    &=
    U (Y ) - \Vert X  - Y  \Vert 
    \\ &\subseteq U (Y ) - \mB \Vert X  - Y  \Vert 
    \\ &= U (Y ) + \mB \Vert X  - Y  \Vert \,.
\end{align*}
Interchanging the roles of $X $ and $Y $, yields $  U (Y ) \subseteq U (X ) + \Vert X  - Y  \Vert  \;\mB$. 
Thus, $\U$ is Lipschitz continuous.
\end{proof} 

\begin{example}\label{ex:robust-optimization}
Robust optimization provides a direct interpretation of the continuity results in this section. Fix an uncertainty set \(\u\), a risk measure \(\rho\), and write \(R:=R^{\u,\rho}\). Let \(A\) be a compact metric set of admissible actions, let \(X\in(\lp{})^d\) be a vector of random factors, and let \(\ell:A\times\mathbb R^d\to\mathbb R\) be a loss function such that \(\ell(a,X)\in\lp{}\) for every \(a\in A\). Since we use the P\&L convention, the random position generated by \(a\) is \(-\ell(a,X)\). Thus, the usual risk-neutral stochastic program \(\min_{a\in A}\E[\ell(a,X)]\) may be replaced by
\[
\min_{a\in A}R(-\ell(a,X)).
\] Consider now perturbed random factors \(X_n\in(\lp{})^d\) and the corresponding problems
\[
\min_{a\in A}R(-\ell(a,X_n)).
\]
Assume that the maps \(a\mapsto R(-\ell(a,X))\) and \(a\mapsto R(-\ell(a,X_n))\) are continuous on \(A\). If \(\u\) is Lipschitz continuous, then \zcref[S]{prop: continuity} implies that \(R\) is Lipschitz continuous; also, by \zcref[S]{theo continuity}, Lipschitz continuity of $R$ can be inherited from \Uset. Hence, for every \(a\in A\),
\[
\bigl|R(-\ell(a,X_n))-R(-\ell(a,X))\bigr|
\leq
\Vert \ell(a,X_n)-\ell(a,X)\Vert .
\]
Therefore, if \(\sup_{a\in A}\Vert \ell(a,X_n)-\ell(a,X)\Vert\to0\), then the objective functions converge uniformly on \(A\). In particular,
\[
\min_{a\in A}R(-\ell(a,X_n))
\longrightarrow
\min_{a\in A}R(-\ell(a,X)).
\] The optimal actions are stable as well. Indeed, define
\[
\Gamma_n:=\arg\min_{a\in A}R(-\ell(a,X_n)),
\qquad
\Gamma:=\arg\min_{a\in A}R(-\ell(a,X)).
\]
By compactness of \(A\), continuity of the objective functions, and uniform convergence, one has
\[
\Limsup_{n\to\infty}\Gamma_n\subseteq\Gamma .
\]
Consequently, if \(\Gamma=\{a^*\}\), then every selection \(a_n^*\in\Gamma_n\) satisfies \(a_n^*\to a^*\).
\end{example}

\section{Dual representation of robust risk measure.}\label{sec:dual-risk}
We recall the classical dual representation of a convex risk measure.  For this, we denote by $ba$  the space of all signed finitely additive measures on $\Omega$, that have bounded total variation and are absolutely continuous with respect to $\P$. Moreover we denote $ba_{1,+} := \left\{ Q \in ba : Q \geq 0, \ Q(\Omega) = 1 \right\}$ the subset of non-negative measures that are normalised to 1. With slight abuse of notation, we write $\E_{Q} [X]$, $Q \in ba $, $ X \in L^\infty$, for the bilinear form between $L^\infty $ and $ba$. We further denote the set of probability measures that are absolutely continuous with respect to $\P$ by $\mathcal{P} \subseteq \ba,$ and denote its members as $\Q \in \mathcal{P}.$ A convex risk measure $\rho$ has representation
\begin{equation}\label{eq:dual-classic}
    \rho(X)= \sup_{Q\in \ba } \{\E_Q[-X] - \alpha_\rho(Q)\}  \,,
\end{equation} 
for some function $\alpha_\rho:ba \rightarrow \R\cup\{\infty\}$. The function $\alpha_\rho$ is commonly referred to as the \emph{penalty function} of $\rho$. Here we also use the terminology that ``$\rho$ is represented over $\alpha_\rho$''. The \emph{minimal penalty function} of $\rho$ is  the Fenchel-Legendre conjugate which we denote by 
$\alpha^{\min}_\rho(Q) := \sup_{ X \in A^\rho } \E_Q[-X]$, see for example Theorem 4.16 in \cite{follmer2025stochastic}. For the base convex risk measure $\rho$, we denote $\alpha_\rho$ for the minimal penalty function in order to ease notation. We also define its proper domain as $\mathcal{Q}_\rho = \{ Q\in ba_{1,+} : \alpha_\rho (Q) < \infty\}$. Moreover, if a convex risk measure is Fatou lsc, then the supremum can be taken over $\mathcal{P}$, which are probability measures on
$(\Omega,\mathcal F)$ that are absolutely continuous w.r.t $\P$ (see e.g., Theorem 4.33 in \cite{follmer2025stochastic}), that is 
\begin{equation*}
    \rho(X)= \sup_{Q\in \mathcal{P} } \{\E_Q[-X] - \alpha_\rho(Q)\}  \,.
\end{equation*}

For the dual representation of robust risk measure we restrict our attention to  convex uncertainty sets. 
 With those assumptions, the robust risk measures are finite, convex, monotone and translation invariant. Moreover the corresponding \Uset~is a convex \Uset, see \zcref[S]{prop:u_to_R,theo:equiv_R_U}. This collection of properties are exactly those needed for the subsequent results on the dual representations.

The building blocks for the dual representation are worst-case expectations, which are the support functions of the uncertainty sets and used as  auxiliary maps.
 
 \begin{definition}\label{defi auxiliar}
     For each $Q\in \ba$ and fixed uncertainty set $u$, the auxiliary map $g_Q^u:\lp{} \rightarrow \R$ is defined as
    \begin{equation*}
        g^u_Q(X)=\sup_{Z\in u(X)}\E_Q[-Z],\qquad\forall\:X\in L^\infty.
    \end{equation*}
    Whenever $u$ is clear from the context, we simply write $g_Q:= g_Q^u$.
 \end{definition}

Note that whenever $u$ is a  convex uncertainty set we have that $g_Q^u$ is a convex risk measure for any $Q \in \ba$ with a dual representation as
\begin{equation*}
g_Q^u (X) = \sup_{Q^\dagger \in \ba}\big\{ \E_{Q^\dagger}[-X ] -\alpha_{g_Q^u}(Q^\dagger) \big\}\,,    
\end{equation*}
where $\alpha_{g_Q^u}$ is the penalty function of $g_Q^u$, that can be written as $\alpha_{g_Q^u}(Q^\dagger) = \sup_{X\in L^\infty} \{ \E_{Q^\dagger} [-X] - g_Q^u(X) \} $.

\begin{lemma}\label{lemma penalty g}
Let $u$ be a 
convex uncertainty set, $R$ a convex robust risk measure, and $U$ be the associated convex consolidated uncertainty set. Then for any $P \in \mathcal{Q}_\rho$, we have that 
$g_P^U(X) = R(X) + \alpha_\rho(P)$ for all $ X \in L^\infty
$, and 
$\alpha_{g_P^U}(Q) = \alpha_R^{\min}(Q) - \alpha_\rho(P)$ for all $Q \in ba_{1,+}$.
\end{lemma}

 \begin{proof}
    Let $P \in\mathcal{Q}_\rho$ and we write $g_P= g_P^U $. By definition, the consolidated uncertainty set is $U(X) = \{Z \in L^\infty : \rho(Z) \le R(X)\} = A^\rho - R(X)$. Thus, the auxiliary map is given by:
\[
g_P(X) = \sup_{Z \in U(X)} \mathbb{E}_P[-Z] = \sup_{Z \in A^\rho - R(X)} \mathbb{E}_P[-Z] = \sup_{Z \in A^\rho} \mathbb{E}_P[-Z] + R(X) =R(X)+ \alpha_{\rho} (P),
\]
as $  \alpha_{\rho} (P) = \sup_{Z \in A^\rho} \mathbb{E}_P[-Z] $. Moreover, we compute the penalty term for $g_P$ evaluated at $Q \in ba_{1,+}$:
\[
\alpha_{g_P}(Q) = \sup_{X \in L^\infty} \{\mathbb{E}_Q[-X] - g_P(X)\} = \sup_{X \in L^\infty} \{\mathbb{E}_Q[-X] - R(X) - \alpha_\rho(P)\}.
\]
Using the fact that the minimal penalty of $R$ can be written as $\alpha_R^{\min}(Q) = \sup_{X \in L^\infty} \{\mathbb{E}_Q[-X] - R(X)\}$, we substitute this to obtain
$\alpha_{g_P}(Q) = \alpha_R^{\min}(Q) - \alpha_\rho(P)$.
\end{proof}

 We are ready to state the main result of this section, the dual representation of robust risk measures.

\begin{theorem}
\label{thm:WC}
Let $u$ be a  convex uncertainty set. Then, $R$ has representation
\begin{align*}
    R(X) =\sup\limits_{Q\in ba_{1,+}}\left\lbrace \E_{Q}[-X]- \alpha_R(Q) \right\rbrace.
\end{align*}
where the penalty function $ \alpha_R $ is given by
\begin{equation}\label{eq:pen}
    \alpha_{R}(Q)=\inf\limits_{Q_2\in ba_{1,+}}\left\lbrace \alpha_\rho(Q_2)+\alpha_{g_{Q_2}}(Q)\right\rbrace,\qquad\forall\:Q\in ba_{1,+}.
\end{equation} 
If $u=U$, then $\alpha_R(P) = \alpha_R^{\min}(P)$ for any $P \in \mathcal{Q}_\rho$.
Moreover, if $Q\mapsto g_Q(X)$ is weak* usc for any $X\in L^\infty$, then the infimum is attained and $\alpha_R$ is weak* lsc. 
\end{theorem}

 \begin{proof}
Under the hypothesis, both the auxiliary map $g_Q$ and $R$ are convex risk measures. Since $\rho$ is a convex risk measure, it is well known (see \cite{follmer2002convex}) that we can write $\rho (X) = \sup_{Q \in \ba} \{ \E_Q [-X] - \alpha_\rho (Q) \}  $, where $\alpha_\rho$ is the penalty function or Fenchel conjugate (see \eqref{eq:dual-classic}); the same applies for $R$ and $g_Q$. Therefore, by  direct calculation we obtain  
\begin{align*}
   R(X)
   &=\sup\limits_{Z\in u(X)}\sup\limits_{Q_1\in ba_{1,+}}\left\lbrace \E_{Q_1}[-Z]-\alpha_\rho(Q_1)\right\rbrace\\
      &=\sup\limits_{Q_1\in ba_{1,+}}\sup\limits_{Z\in u(X)}\left\lbrace \E_{Q_1}[-Z]-\alpha_\rho(Q_1)\right\rbrace\\
    &=\sup\limits_{Q_1\in ba_{1,+}}\left\lbrace g_{Q_1}(X)-\alpha_\rho(Q_1)\right\rbrace\\
    &=\sup\limits_{Q_1\in ba_{1,+}}\left\lbrace  \sup\limits_{Q_2\in ba_{1,+}}\left\lbrace \E_{Q_2}[-X] -\alpha_{g_{Q_1}}(Q_2) \right\rbrace -\alpha_\rho(Q_1)\right\rbrace\\
    &=\sup\limits_{Q_1,Q_2\in ba_{1,+}}\left\lbrace \E_{Q_2}[-X]-\alpha_\rho(Q_1)-\alpha_{g_{Q_1}}(Q_2)\right\rbrace\\
    &=\sup\limits_{Q_1\in ba_{1,+}}\left\lbrace \E_{Q_1}[-X]-\inf\limits_{Q_2\in ba_{1,+}}\left\lbrace\alpha_\rho(Q_2)+\alpha_{g_{Q_2}}(Q_1)\right\rbrace\right\rbrace.
\end{align*} Since the same deduction holds by using consolidated uncertainty sets $U$ instead of the base $u$, we can safely consider $u=U$. In this case, we have that $R$ can also be represented by $Q\mapsto\inf\limits_{Q_2\in ba_{1,+}}\left\lbrace \alpha_\rho(Q_2)+\alpha_{g^U_{Q_2}}(Q)\right\rbrace$. Thus, from \zcref[S]{lemma penalty g}, for any $Q\in\mathcal{Q}_\rho$, substituting  $\alpha_{g^U_P}$ back into the obtained representation yields:
\[
\alpha_R(Q) = \inf_{P \in ba_{1,+}} \{\alpha_\rho(P) + \alpha_{g_P}(Q)\} = \inf_{P \in ba_{1,+}} \{\alpha_\rho(P) + \alpha_R^{\min}(Q) - \alpha_\rho(P)\} = \alpha_R^{\min}(Q).
\]
Thus, the obtained penalty coincides with the minimal penalty on $\mathcal{Q}_\rho$.

Now, let $Q\mapsto g_Q(X)$ be weak* usc. Thus, $Q_2\mapsto\alpha_\rho(Q_2)+\alpha_{g_{Q_2}}(Q_1)$ is weak* lsc for any fixed $Q_1\in\ba$. Then the infimum is attained since it is taken over the weak* compact $ba_{1,+}$. In order to verify that $\alpha_R$ is weak* lsc, we let $c\in\mathbb{R}$ and note that its lower level set at $c$ is given as
\[A_c=\{Q\in ba_{1,+}\colon\exists\:Q_2\in ba_{1,+}\quad \text{s.t.}\quad \alpha_\rho(Q_2)+\alpha_{g_{Q_2}}(Q)\leq c\}.\] We claim $A_c$ is weak* closed. Let $(Q_\gamma)_{\gamma \in \Gamma} \subseteq A_c$ be a net such that $Q_\gamma \xrightarrow{w^*} Q$ in $\sigma(ba, L^\infty)$. By the definition of $A_c$, for each $\gamma \in \Gamma$, there exists $Q_{2,\gamma} \in ba_{1,+}$ such that $\alpha_\rho(Q_{2,\gamma}) + \alpha_{g_{Q_{2,\gamma}}}(Q_\gamma) \le c$. 
Since $ba_{1,+}$ is weak* compact by the Banach-Alaoglu theorem, the net $(Q_{2,\gamma})_{\gamma \in \Gamma}$ has a convergent subnet $(Q_{2,\gamma_i})_{i \in I}$ that converges to some $Q_2 \in ba_{1,+}$ in the weak* topology. Note that for a fixed $X \in L^\infty$, the map $(P_1, P_2) \mapsto \mathbb{E}_{P_1}[-X] + \alpha_\rho(P_2) - g_{P_2}(X)$ is jointly weak* lsc on $ba_{1,+} \times ba_{1,+}$. Because the supremum of a family of lsc functions is lsc, the map 
$$ (P_1, P_2) \mapsto \sup_{X \in L^\infty} \{\mathbb{E}_{P_1}[-X] + \alpha_\rho(P_2) - g_{P_2}(X)\} = \alpha_\rho(P_2) + \alpha_{g_{P_2}}(P_1) $$
is also jointly weak* lsc. Therefore, evaluating the limit along the convergent subnet yields:
$$ \alpha_R(Q) \le \alpha_\rho(Q_2) + \alpha_{g_{Q_2}}(Q) \le \liminf_{i \in I} \left( \alpha_\rho(Q_{2,\gamma_i}) + \alpha_{g_{Q_{2,\gamma_i}}}(Q_{\gamma_i}) \right) \le c. $$
Hence, $Q \in A_c$. Since $c \in \mathbb{R}$ was arbitrary, $\alpha_R$ is weak* lower semi-continuous.
\end{proof}

\begin{corollary}
If in addition to the hypotheses in \zcref[S]{thm:WC}, $\rho$ is also positive homogeneous, then the penalty becomes
\begin{align*}
     \alpha_{R}(Q)=\inf\left\lbrace \alpha_{g_{Q_2}}(Q) : {Q_2\in \mathcal{Q}_\rho}\right\rbrace,\qquad\forall\:Q\in ba_{1,+},
\end{align*}
where $\mathcal{Q}_\rho := \{ Q \in \ba : \alpha_\rho (Q) < \infty\}$.
If further $u$ is positive homogeneous, then the 
dual representation becomes 
\begin{align*}
    R(X) =\sup\limits_{Q\in \q_R }\E_{Q}[-X], \qquad\forall\:X\in \lp{},
\end{align*}
where $\mathcal{Q}_R := \{ Q \in \ba : \alpha_R (Q) < \infty\}$.
\end{corollary}

 \begin{proof}
The first claim follows directly from noting that if $\rho$ is coherent $\alpha_\rho$ assumes either $\infty$ or $0$. And the second claim follows from the same reasoning, as $R$ is a coherent risk measure. 
\end{proof}

If the risk measure and the \Uset~are Fatou lsc, then the dual representation is refined.

\begin{corollary}
    \label{crl:probability-representation}
Let $u$ be a convex uncertainty set and suppose in addition,   that
the convex risk measure $\rho$ and $u$ are Fatou lsc. Then the robust
risk measure $R$ admits a dual representation over countably additive
probabilities:
\[
  R(X)
  \;=\;
  \sup_{\Q\in\mathcal P}
  \Bigl\{ \E_{\Q}[-X] - \alpha_R(\Q)\Bigr\},
  \qquad X\in L^\infty.
\]
The penalty function $\alpha_R$ can be chosen as
\[
  \alpha_R(\Q)
  \;=\;
  \inf_{\Q^\dagger\in\mathcal P}
  \Bigl\{
    \alpha_\rho(\Q^\dagger) + \alpha_{g_{\Q^\dagger}}(\Q)
  \Bigr\},
  \qquad \Q\in\mathcal P.
\]
\end{corollary}

 \begin{proof}
 By assumption, $\rho$ is a convex monetary risk measure with the Fatou
property. Hence, by the standard dual representation theorem on
$(L^\infty,L^1)$, $\rho$ admits a representation over countably additive
probabilities. That is there exists a penalty $\alpha_\rho:\mathcal P\to[0,+\infty]$
such that
\[
  \rho(Z)
  \;=\;
  \sup_{\Q\in\mathcal P}
  \bigl\{ \E_{\Q}[-Z] - \alpha_\rho(\Q)\bigr\},
  \qquad Z\in L^\infty.
\]
Under the assumptions on $u$ and $\rho$, \zcref[S]{prop: continuity} and \zcref[S]{theo continuity}
imply that the robust risk measure $R$, the consolidated uncertainty set $U$, and all
auxiliary maps $g_{\Q^\dagger}$, $\Q^\dagger\in\mathcal P$ are Fatou lsc.
Applying the
same duality result to $R$ and to each $g_{\Q^\dagger}$ provides dual representations
over $\mathcal P$.

Repeating the steps of the proof of \zcref[S]{thm:WC}, but restricting
all dual variables to $\mathcal P$, we obtain for every $X\in L^\infty$:
\[
  R(X)
  =
  \sup_{\Q\in\mathcal P}
  \biggl\{
    \E_{\Q}[-X]
    -
    \inf_{\Q^\dagger\in\mathcal P}
    \bigl[
      \alpha_\rho(\Q^\dagger)+\alpha_{g_{\Q^\dagger}}(\Q)
    \bigr]
  \biggr\}.
\]
Thus, we may choose
\[
  \alpha_R(\Q)
  :=
  \inf_{\Q^\dagger\in\mathcal P}
  \bigl\{
    \alpha_\rho(\Q^\dagger)+\alpha_{g_{\Q^\dagger}}(\Q)
  \bigr\},
  \qquad \Q\in\mathcal P.
\]
\end{proof}

We next provide an example of commonly used uncertainty sets in the literature and explicitly calculate the dual representation of the resulting robust risk measures and their penalty functions.

\example{Let $u(X)=\{Y\in L^\infty\colon \lVert X-Y\rVert\leq \epsilon\}=X+\{Y\in L^\infty\colon \lVert Y\rVert\leq\epsilon\}$ be closed balls in $L^\infty$. In this case, for any $Q_2\in ba_{1,+}$ we have that $g_{Q_2}(X)=\E_{Q_2}[-X]+\epsilon$. Thus, we get that 
\[\alpha_{g_{Q_2}}(Q_1)=\sup\limits_{X\in L^\infty}\left\lbrace\E_{Q_1}[-X]-\E_{Q_2}[-X] \right\rbrace-\epsilon=\chi_{Q_1=Q_2}-\epsilon,\]
where $\chi$ is defined as  $\chi_{Q_1=Q_2}=\infty$ if $Q_1\neq Q_2$, and  $\chi_{Q_1=Q_2}=0$ otherwise. By \zcref[S]{thm:WC}, $R$ has representation 
\[\alpha_R(Q_1)= \min\limits_{Q_2\in ba_{1,+}}\left\lbrace \alpha_\rho(Q_2)+\chi_{Q_1=Q_2}-\epsilon\right\rbrace=\alpha_\rho(Q_1)-\epsilon.\] In this case $R(X)=\rho(X)+\epsilon$. If our domain is $L^p$, $p\in [1,\infty)$, with dual $L^q$ where $\frac{1}{p}+\frac{1}{q}=1$, then the same reasoning applies and we obtain $\alpha_R(Q_1)=\alpha_\rho(Q_1)-\epsilon\lVert \frac{dQ_1}{d\mathbb{P}}\rVert_q$. This is in consonance with the work of \cite{Righi2024}}.

\section{
Dual representation of uncertainty Sets.}\label{sec:dual-U}
In this section, we derive the dual representation of uncertainty sets. As uncertainty sets are set-valued maps, we require the set-valued conjugation framework in the sense of Hamel \cite{hamel2009duality}. Thus, we first recall the necessary notation and definition for this setup.

\subsection{Notations and definitions.}
Throughout this section, without loss of generality, we shall require $\rho$ to be normalised. We use the theory developed in \cite{hamel2009duality,hamel2015set} which fits neatly with the notion of \Uset s.
Let $\X := \{A\subseteq L^\infty : A = {\clconv}(A+L^\infty_+)  \}$. Thus, if $U$ is a \Uset\,of an order preserving uncertainty set $u$, then $U:L^\infty \rightarrow \mathcal{X}$, that is $U$ maps to the set $\mathcal{X}$. Moreover, if $A \in \mathcal{X}$, then there exists a \Uset\, $U$, such that $U(X) = A$ for some $X\in L^\infty$. Additionally, $U(X) = U(X) + L^\infty_+$ for all $X\in L^\infty$.

The space $(\mathcal{X},\subseteq)$ is a lattice with $\sup \mathcal{B}=\bigcap_{B\in\mathcal{B}}B$ and $\inf \mathcal{B}=\clconv\left(\bigcup_{B\in\mathcal{B}}B\right)$, see Theorem 6 of \cite{hamel2005variational}. Due to order preservingness, the \Uset\, is monotone, i.e. $X\leq Y$ implies $U(Y) \subseteq U(X) = U(X) + L^\infty_+$. This naturally gives rise to a pre-order on $\mathcal{X}$ defined by $U \le U'$ if and only if $U' \subseteq U + L^\infty_+ = U$.
We denote by $\epi$ the epigraph of a \Uset\, $U$, that is
\begin{align*}
    \epi &:=\{ (X,Z) \in L^\infty \times L^\infty\colon U(X) \leq \{Z\}\} \\
    &= \{ (X,Z) \in L^\infty \times L^\infty\colon \{Z\} \subseteq U(X) + L^\infty_+\}  
    \\ &= \{ (X,Z) \in L^\infty \times L^\infty\colon \exists\:Y \in U(X)\:\text{s.t.}\:Y\leq Z\}.  
\end{align*}

The support of the graph of $U$ is a function $\delta_U : ba\times ba \rightarrow \R \cup\{\infty\}$ defined as
\begin{align}\label{eq:delta-U-def}
    \delta_U (Q_1,Q_2) = \sup\{ \E_{Q_1}[X]+\E_{Q_2}[Z] :X\in L^\infty,\, Z\in U(X) \}\,.
\end{align}
We further define the non-positive hyperplane of $Q \in ba$ as $H(Q) := \{Z \in L^\infty : \E_Q [Z] \leq 0\} $ and the negative polar of $L^\infty_+$ (without the $0$-element) as $ba_-: =(L^\infty_+)^*=\{{Q}\in ba\setminus\{0\} \colon \E_{Q}(Z)\leq 0,\:\forall\:Z\in L^\infty_+\}$.  Notice that $H(Q)\in \mathcal{X}$ for any $Q\in ba_-$.

We take, as building blocks for the conjugate function, the co-linear maps that are $\mathcal{X}$-valued as
 $S_{{Q}_1,{Q}_2}\colon L^\infty\to 2^{L^\infty}$, defined for each $({Q}_1,{Q}_2)\in ba\times ba_-$ as 
 \begin{equation*}
    S_{{Q}_1,{Q}_2}(X)=\left\lbrace Z\in L^\infty\colon \E_{{Q}_1}[X] + \E_{{Q}_2}[Z]\leq0\right\rbrace=H({Q}_2)+ \dfrac{\E_{{Q}_1}[X]}{-Q_2(\Omega)}\,.
\end{equation*}
We highlight that by definition $Q_2(\Omega)<0$. Armed with this notation, we can define the conjugate  and bi-conjugate  of the uncertainty set  $u$. With slight abuse of notation we denote the set-valued convex conjugate by $\alpha_U$, mirroring the standard notation for the scalar penalty function of risk measures. 

\begin{definition}\label{def:dual-u}
    Let $U : L^\infty \to 2^{L^\infty}$ be a \Uset. Then 
    $\alpha_U\colon ba\times ba_-\to 2^{L^\infty}$ defined as 
    \begin{equation*}
\alpha_U(Q_1,Q_2):=U^*(Q_1,Q_2)= cl\bigcup\limits_{X\in L^\infty}\left\lbrace U(X)+ S_{Q_1,Q_2}(-X)\right\rbrace
\end{equation*}
is the convex conjugate of $U$.  Furthermore, the function $U^{**}:L^\infty \to 2^{L^\infty}$ defined by
\begin{equation*}
U^{**}(X)=\bigcap\limits_{({Q}_1,{Q}_2)\in ba\times ba_-}  \left\lbrace \alpha_U({Q}_1,{Q}_2)+S_{{Q}_1,{Q}_2}(X) \right\rbrace, 
\end{equation*}
is the bi-conjugate of $U$.
\end{definition} 

A direct adaptation of Lemma 1 of \cite{hamel2009duality} allow us to rewrite the convex conjugate as \[\alpha_U(Q_1,Q_2)=\{ Z\in L^\infty\colon \E_{Q_2}[Z]\leq \delta_{U}(Q_1,Q_2)\},\] and the bi-conjugate as \[
U^{**}(X)=\bigcap\limits_{({Q}_1,{Q}_2)\in ba\times ba_-}  \{ Z\in L^\infty\colon \E_{Q_1}[X] +  \E_{Q_2}[Z]\leq \delta_{U}(Q_1,Q_2)\}.\] 

Next, we establish when the bi-conjugate of \Uset~equals itself.

\begin{lemma}\label{lmm:U  bi}
    Let $U$ be a \Uset. Then $U = U^{**}$ if and only if $U$ is set-concave and its associated robust risk measure $R$ is usc.
\end{lemma}

 \begin{proof}

Theorem~2 of Hamel (2009) ensures that $U = U^{**}$ if and only if, besides the image of $U$ lying in $\mathcal{X}$, the map $U$ is set-concave and closed in $(\mathcal{X}, \subseteq)$, that is, if its epigraph is closed. Equivalently, this requires that
$U(X) \supseteq \limsup_{n \to \infty} U(X^n).$
By \zcref[S]{theo continuity}, this inclusion is satisfied if and only if $R$ is upper semi-continuous.
\end{proof}

\subsection{Duality of set-valued maps}

Now, we state the main result of this section, the dual representation of \Uset s.
    
\begin{theorem}\label{theorem dual U}
     $U$ is a concave \Uset\,  if and only if it can be represented as \begin{equation}\label{eq:U-dual}
 U(X)=U^{**}(X)=\bigcap\limits_{{Q}_1,{Q}_2\in \ba}   \left\lbrace \alpha_U({Q}_1,-{Q}_2)+S_{{Q}_1,-{Q}_2}(X) \right\rbrace.
\end{equation}
where  the penalty term is defined as 
\begin{align}\label{teorem dual U: eq penalty}
\alpha_U({Q}_1,{-Q}_2)
=\begin{cases}
    H(-Q_2)-\alpha_\rho(Q_2)-\beta_R(Q_1),&\:\beta_R(Q_1),\alpha_\rho(Q_2)<\infty\\[0.5em]
    L^\infty,& \text{otherwise},
\end{cases}
\end{align}
where $\beta_R(Q)=\sup_{X\in \lp{}}\{ \E_Q[X] +R(X)\}$ is proper.

\end{theorem}

 \begin{proof}
For the if part, let $V$ be defined as the right-hand-side of \eqref{eq:U-dual}, i.e., \[V(X)=\bigcap\limits_{{Q}_1,{Q}_2\in \ba}   \left\lbrace \alpha_U({Q}_1,-{Q}_2)+S_{{Q}_1,-{Q}_2}(X) \right\rbrace.\] Then from Theorem 2 in \cite{hamel2009duality} we have that $V$
 is closed, convex, non-empty and $V(X)+L^\infty_+=V(X)$. 
Moreover, $V$ is set-concave and has a closed epigraph. Then translation invariance follows by 
recalling that $S_{Q_1,-Q_2}(X)=H(-Q_2)+\E_{Q_1}[X](Q_2(\Omega))^{-1}$ and then noting that, for all $c \in  \R$,
\begin{align*}
    V(X+c) &= \bigcap\limits_{{Q}_1,{Q}_2 \in \ba}   \left\lbrace \alpha_U({Q}_1,-{Q}_2)+S_{{Q}_1,-{Q}_2}(X+c) \right\rbrace
  \\    &  =\bigcap\limits_{{Q}_1,{Q}_2\in \ba}   \left\lbrace \alpha_U({Q}_1,-{Q}_2)+H(-Q_2)+\E_{Q_1}[X]  + c \right\rbrace
    \\ &  =\bigcap\limits_{{Q}_1,{Q}_2\in \ba}   \left\lbrace \alpha_U({Q}_1,-{Q}_2)+H(-Q_2)+\E_{Q_1}[X]\right\rbrace +c\\
    &= V(X) + c\,.
\end{align*}
Next we show that $V$ is monotone. Note that for any  $Q \in \ba$ there exists a \Uset \, $V'$ such that $V'(0) = H(-Q)$. Hence, by \zcref[S]{lemma U containted in U + Y}, \zcref[S]{lemma U containted in U + Y: item Y non negative}, it holds  that $H(-Q) +Y\subseteq H(-Q)$ for any $Y \in \lp{+}_+$. Therefore,  taking $X\leq Z$ and $Y \in \lp{+}_+$ such that $Z=X+Y$ yields
\begin{align*}
    V(Z) &=    \bigcap\limits_{{Q}_1,{Q}_2 \in \ba}   \left\lbrace \alpha_U({Q}_1,-{Q}_2)+S_{{Q}_1,-{Q}_2}(Z) \right\rbrace
    \\ &  =\bigcap\limits_{{Q}_1,{Q}_2\in \ba}   \left\lbrace \alpha_U({Q}_1,-{Q}_2)+H(-Q_2)+\E_{Q_1}[Y] +\E_{Q_1}[X]   \right\rbrace.
    \\ &  \subseteq \bigcap\limits_{{Q}_1,{Q}_2\in \ba}   \left\lbrace \alpha_U({Q}_1,-{Q}_2)+H(-Q_2)+\E_{Q_1}[X]\right\rbrace
    \\ &= \bigcap\limits_{{Q}_1,{Q}_2\in \ba}   \left\lbrace \alpha_U({Q}_1,-{Q}_2)+S_{{Q}_1,-{Q}_2}(X) \right\rbrace
    \\
    &= V(X).
\end{align*}
Thus, $V(X)$ is a monotone, translation invariant, and set-concave \Uset~associated with the risk measure $\rho_V (X) := \inf \{m \in \R : X+m \in V(0)  \}$ and the robust risk measure $R(X) = \sup_{Z \in V(X)} \rho_V (Z)$.

Now we turn to the only if part. Since for each $X \in \lp{}$, $U(X)$ is a sub-level set of a convex  $\rho$, it is convex, closed, and  $U(X) = U(X) + L^\infty_+$. Hence, the image of $U$ lies in $\X$, that is, we can take $U:\lp{} \rightarrow \X$. Further, by \zcref[S]{lemma lipschitz}, it is Lipschitz continuous. Thus, its epigraph is closed. Thus, by \zcref[S]{lmm:U  bi} we have $U = U^{**}$ if and only if $U$ is set-concave, where the latter holds by assumption. Therefore we have \[U(X) = U^{**}(X) = \bigcap\limits_{({Q}_1,{Q}_2)\in ba\times ba_-}  \left\lbrace \alpha_U({Q}_1,{Q}_2)+S_{{Q}_1,{Q}_2}(X) \right\rbrace,\] where the second equality follows by \zcref[S]{def:dual-u}. In the next step we show that we can take the intersection over $\ba$ instead of $ba$. It follows from the definition of $S_{Q_1,Q_2}$ for $(Q_1,Q_2) \in ba \times ba_- $ that 
$S_{Q_1,Q_2}(X+Y) = S_{Q_1,Q_2}(X) + \frac{\E_{Q_1}[Y]}{-Q_2(\Omega)}$, in particular, 
    $S_{Q_1,Q_2}(X+c)= S_{Q_1,Q_2}(X) +c \frac{Q_1(\Omega)}{-Q_2(\Omega)}$ for some constant $c$. 
Then, as $U$ is translation invariant, it holds that
\begin{align*}
    U(X)+c=U(X+c)&=\bigcap\limits_{({Q}_1,{Q}_2)\in ba\times ba_-}  \left\lbrace \alpha_U({Q}_1,{Q}_2)+S_{{Q}_1,{Q}_2}(X+c) \right\rbrace 
    \\ 
&=\bigcap\limits_{({Q}_1,{Q}_2)\in ba\times ba_-}  \left\lbrace \alpha_U({Q}_1,{Q}_2)+S_{{Q}_1,{Q}_2}(X)  +c\frac{Q_1(\Omega)}{-Q_2(\Omega)} \right\rbrace. \end{align*}
From this, it follows that, for all $({Q}_1,{Q}_2)\in ba\times ba_-$ and $c \in \R$ that 
\begin{align*}
    U(X) + c = U(X+c) \subseteq  \alpha_U({Q}_1,{Q}_2)+S_{{Q}_1,{Q}_2}(X)  +c\frac{Q_1(\Omega)}{-Q_2(\Omega)},
\end{align*}
which is equivalent to
$
    U(X)  +c(1-\frac{Q_1(\Omega)}{-Q_2(\Omega)}) \subseteq  \alpha_U({Q}_1 ,{Q}_2)+S_{{Q}_1,{Q}_2}(X) 
$. If $\frac{Q_1(\Omega)}{-Q_2(\Omega)}\neq 1$, 
taking the union over $c \in \R$ yields
\begin{align*}
    \lp{} = U(X) +\R= \bigcup_{c \in \R}\big\{U(X)+c\big(1-\tfrac{Q_1(\Omega)}{-Q_2(\Omega)}\big)\big\} \subseteq  \alpha_U({Q}_1 ,{Q}_2)+S_{{Q}_1,{Q}_2}(X).
\end{align*}
In this case, restricting the intersection only for $\frac{Q_1(\Omega)}{-Q_2(\Omega)}=  1$ does not change it, and 
\begin{align*}
    U(X) &= 
    \bigcap_{ ({Q}_1,{Q}_2)\in ba\times ba_-} \bigg\{\alpha_U({Q}_1 ,{Q}_2)+S_{{Q}_1,{Q}_2}(X)\bigg\} 
    \\  &=  \bigcap_{ ({Q}_1,{Q}_2)\in ba\times ba_-} \bigg\{\alpha_U({Q}_1 ,{Q}_2)+S_{{Q}_1,{Q}_2}(X) : Q_1(\Omega) =-Q_2(\Omega)\bigg\}
    \\  &=  \bigcap_{ ({Q}_1,{Q}_2)\in ba\times ba_-} \bigg\{\alpha_U({Q}_1 ,{Q}_2)+S_{{Q}_1,{Q}_2}(X) : Q_1(\Omega) =-Q_2(\Omega) = 1\bigg\} 
    \\  &=  \bigcap_{ {Q}_1\in ba} \bigg\{\alpha_U({Q}_1 ,-{Q}_2)+S_{{Q}_1,-{Q}_2}(X) :  {Q}_2\in  ba_+ \setminus \{0\} ,Q_1(\Omega) =Q_2(\Omega) = 1\bigg\} .
\end{align*}
The third equality follows as $S_{\lambda Q_1, \lambda Q_2}(X) = S_{Q_1,Q_2}(X)$ for any $0<\lambda \in \R $, and the same holds for $\alpha_U$. Hence, we can take the normalised $Q_1' = \frac{Q_1}{Q_1(\Omega)}$ and $Q_2' = \frac{Q_2}{Q_1(\Omega)}$.

Now, assume, by contradiction, that $Q_1 \notin ba_+$. Then, there exists $Y \in L^\infty_+$ such that $\E_{Q_1}[Y] < 0$. Let $W_0 \in L^\infty$ be arbitrary and choose any $Z_0 \in U(W_0)$. For any $\lambda > 0$, define $W_\lambda := W_0 - \lambda Y$. Since $Y \in L^\infty_+$, we have $W_\lambda \leq W_0$. By the monotonicity of $U$, it follows that $U(W_0) \subseteq U(W_\lambda)$, which guarantees $Z_0 \in U(W_\lambda)$. Evaluating the support function of the graph of $U$ yields
\begin{align*}
    \delta_U(Q_1, -Q_2) &= \sup_{W \in L^\infty, Z \in U(W)} \{ \E_{Q_1}[W] + \E_{-Q_2}[Z] \} \\
    &\geq \sup_{\lambda > 0} \{ \E_{Q_1}[W_\lambda] + \E_{-Q_2}[Z_0] \} \\
    &= \E_{Q_1}[W_0] + \E_{-Q_2}[Z_0] - \inf_{\lambda > 0} \lambda \E_{Q_1}[Y].
\end{align*}
Since $\E_{Q_1}[Y] < 0$, it follows that $-\lambda \E_{Q_1}[Y] \to +\infty$ as $\lambda \to \infty$. Therefore, $\delta_U(Q_1, -Q_2) = +\infty$. This implies that the penalty set evaluates to the entire space:
\begin{align*}
    \alpha_U(Q_1, -Q_2) &= \{ Z \in L^\infty : \E_{-Q_2}[Z] \leq \delta_U(Q_1, -Q_2) \} = L^\infty.
\end{align*}
Consequently, $\alpha_U(Q_1, -Q_2) + S_{Q_1, -Q_2}(X) = L^\infty$. Since intersecting with $L^\infty$ imposes no additional restriction on $U(X)$, any functional $Q_1$ possessing a negative component can be omitted from the dual representation without altering the intersection. Hence, we can safely restrict the intersection over $Q_1 \in ba_+$, and after normalization, over $\ba \times \ba$, that is, we obtain
\begin{align*}
    U(X) &= 
          \bigcap_{ {Q}_1,{Q}_2\in \ba} \bigg\{\alpha_U({Q}_1 ,-{Q}_2)+S_{{Q}_1,-{Q}_2}(X)\bigg\}.
\end{align*}

For the claim on $\alpha_U$ recall that $U(X)={A}^\rho-R(X)$ and that the concave conjugate   is defined as $\beta_R(Q_1)=\sup_{X\in L^\infty}\{\E_{Q_1}[X]+R(X)\}$.
Then, for \(Q_1,Q_2\in\ba\),
\begin{align*}
\alpha_U(Q_1,-Q_2)
&=cl\bigcup_{X\in L^\infty}
\left\lbrace {A}^\rho+S_{Q_1,-Q_2}(-R(X)-X)\right\rbrace\\
&=cl\left\lbrace W\in L^\infty\colon
\exists X\in L^\infty,\ \exists (Y,Z)\in {A^\rho}\times
S_{Q_1,-Q_2}(-R(X)-X),\ W=Y+Z\right\rbrace\\
&=cl\left\lbrace W\in L^\infty\colon
\exists X\in L^\infty,\ \exists Y\in {A^\rho},\
0\geq \E_{Q_1}[-R(X)-X]+\E_{-Q_2}[W-Y]\right\rbrace\\
&=cl\left\lbrace W\in L^\infty\colon
0\geq \inf_{X\in L^\infty}\{\E_{Q_1}[-X]-R(X)\}
+\E_{-Q_2}[W]
+\inf_{Y\in A^\rho}\E_{-Q_2}[-Y]\right\rbrace\\
&=cl\left\lbrace W\in L^\infty\colon
0\geq -\sup_{X\in L^\infty}\{\E_{Q_1}[X]+R(X)\}
+\E_{-Q_2}[W]
-\sup_{Y\in A^\rho}\E_{Q_2}[-Y]\right\rbrace\\
&=cl\left\lbrace W\in L^\infty\colon
0\geq \E_{-Q_2}[W]-\beta_R(Q_1)-\alpha_\rho(Q_2)\right\rbrace\\
&=H(-Q_2)-\beta_R(Q_1)-\alpha_\rho(Q_2).
\end{align*}
The second last  equalities follow because $\alpha_\rho(Q)=\sup_{X\in{A}^\rho}\E_Q[-X]$. 
Note that as $U(X)$ is bounded from below, for all $X \in \lp{}$ the function $\beta_R$ must be proper, as otherwise $U(X) = \lp{}_+$.
\end{proof}


\begin{remark}
We have that $\beta_R\colon ba \rightarrow \R\cup\{\infty\}$ is the convex conjugate of $-R$ under the bilinear duality form taken as $(X,Q)\mapsto \E[XQ]$. Note that $\beta_R$ is convex and weak* lsc even if $R$ is not convex. The reversion of signs seems to be necessary to properly accommodate the order reversion in the Hamel duality approach. 
\end{remark}

We now show a corollary for the \emph{coherent} case for \Uset s, where the penalty becomes analogous to a convex indicator function.

\begin{corollary}
If in addition to the hypotheses in \zcref[S]{theorem dual U}, $U$ is also positive homogeneous, then the penalty becomes
\[\alpha_U( {Q}_1,- {Q}_2)=\begin{cases}
 H(- {Q}_2),&\:( {Q}_1, {Q}_2)\in\mathcal{Q}_U\\
 L^\infty,&\:( {Q}_1, {Q}_2)\not\in\mathcal{Q}_U
\end{cases},\] 
where
\[\mathcal{Q}_U=\big\{( {Q}_1, {Q}_2)\in \ba\times \ba \colon S_{( {Q}_1,- {Q}_2)}(X)\supseteq U(X)\,,\quad\forall\:X\in L^\infty\big\}.\] In this case the
dual representation becomes 
\[U(X)=\bigcap\limits_{( {Q}_1, {Q}_2)\in \mathcal{Q}_U}  S_{( {Q}_1,- {Q}_2)}(X). \]
\end{corollary}

 \begin{proof}
If $( {Q}_1, {Q}_2)\in\mathcal{Q}_U$, then $U(0)\subseteq S_{( {Q}_1,- {Q}_2)}(0)=H(- {Q}_2)$. Of course, $H(- {Q}_2)\supseteq\alpha_U( {Q}_1,- {Q}_2)\supseteq H(- {Q}_2)$. 
Thus, $\alpha_U( {Q}_1,- {Q}_2)=H(- {Q}_2)$. If $( {Q}_1, {Q}_2)\not\in\mathcal{Q}_U$, then there exist $(X,Z)\in graph_U$ such that $Z\not\in S_{( {Q}_1,- {Q}_2)}(X)$, which implies $\E_{ {Q}_1}[X] + \E_{- {Q}_2}[Z]>0$. Since $U$ is sub-linear, its graph is a convex cone. Then we have that
\[\delta_{graph_U}( {Q}_1, -{Q}_2)\geq \sup\limits_{\lambda\geq 0}\{\E_{ {Q}_1}[\lambda X]+\E_{ -
{Q}_2}[\lambda Z]\}=\infty.\] Thus, $\alpha_U( {Q}_1,- {Q}_2)=L^\infty$. Hence, the claim follows.
\end{proof}

The dual representation of $U$ can be taken over probability measures, as long as the \Uset~satisfies the Fatou usc, the next corollary formalizes that.

\begin{corollary}
Let $U$ be a concave \Uset , and $\rho$ be Fatou lsc. Then  $U$ is Fatou usc if and only if
\begin{equation*}
 U(X)=\bigcap\limits_{{Q}_1,{Q}_2\in \mathcal{P}}   \left\lbrace \alpha_U({\Q}_1,-{\Q}_2)+S_{{\Q}_1,-{\Q}_2}(X) \right\rbrace.
\end{equation*}
\end{corollary}

 \begin{proof}
Assume first that the representation is taken over $\mathcal P$. Let
$\{X^n\}_{n\in\N}\subseteq L^\infty$ be bounded and such that $X^n\limas X$.
Take $Y\in \Limsup_{n\to\infty}U(X^n)$. Then there exist a subsequence
$\{X^{n_k}\}_{k\in\N}$ and $Y^{n_k}\in U(X^{n_k})$ such that $Y^{n_k}\limL Y$.
For every $Q_1,Q_2\in\mathcal P$, since $Y^{n_k}\in U(X^{n_k})$, the
representation gives
\[
\E_{Q_1}[X^{n_k}]+\E_{-Q_2}[Y^{n_k}]
\leq
\delta_U(Q_1,-Q_2).
\]
By dominated convergence, $\E_{Q_1}[X^{n_k}]\to\E_{Q_1}[X]$, and since
$Y^{n_k}\limL Y$, also $\E_{-Q_2}[Y^{n_k}]\to\E_{-Q_2}[Y]$. Therefore
$\E_{Q_1}[X]+\E_{-Q_2}[Y]\leq \delta_U(Q_1,-Q_2)$ for all
$Q_1,Q_2\in\mathcal P$. Using again the representation over $\mathcal P$, we
obtain $Y\in U(X)$. Hence, $\Limsup_{n\to\infty}U(X^n)\subseteq U(X)$, and $U$
is Fatou usc.

Conversely, assume that \(U\) is Fatou usc. By \zcref[S]{theo continuity}, \(R\) is Fatou usc. We show that the representation can be taken over the dual pair $(L^\infty\times L^\infty,\; L^1\times L^1).$
Let $\epi    :=
    \left\{(X,Z)\in L^\infty\times L^\infty : Z\in U(X)\right\}$. Since \(U\) is set-concave, \(\epi \) is convex. We claim that
\(\epi \) is closed under bounded a.s. convergence. Let
\(\{(X^n,Z^n)\}_{n\in\mathbb N}\subseteq \epi \) be uniformly bounded in
\(L^\infty\times L^\infty\) and suppose that $X^n \limas X$ and  $   Z^n \limas Z$.
Since \((X^n,Z^n)\in \epi \), we have \(Z^n\in U(X^n)\). Equivalently,
\[
    \rho(Z^n)\leq R(X^n), \qquad n\in\mathbb N.
\]
As \(\rho\) is Fatou lsc and \(\{Z^n\}_{n\in\mathbb N}\) is bounded, it follows that
\[
    \rho(Z)\leq \liminf_{n\to\infty}\rho(Z^n).
\]
Moreover, since \(R\) is Fatou usc and \(\{X^n\}_{n\in\mathbb N}\) is bounded,
\[
    \limsup_{n\to\infty}R(X^n)\leq R(X).
\]
Therefore,
\[
    \rho(Z)
    \leq \liminf_{n\to\infty}\rho(Z^n)
    \leq \liminf_{n\to\infty}R(X^n)
    \leq \limsup_{n\to\infty}R(X^n)
    \leq R(X).
\]
Hence \(Z\in U(X)\), and consequently \((X,Z)\in \epi \). Thus,
\(\epi \) is closed under bounded a.s. convergence.

Since \(\epi \) is convex and Fatou closed, the Fatou closedness
criterion for convex subsets of \(L^\infty\times L^\infty\) implies that
\(\epi \) is
\(\sigma(L^\infty\times L^\infty,L^1\times L^1)\)-closed. Hence the biconjugation
argument for \(U\) can be applied with respect to the dual pair $ (L^\infty\times L^\infty,\; L^1\times L^1),$ rather than
$  (L^\infty\times L^\infty,\; ba\times ba)$
Therefore, the dual representation of \(U\) can be written using dual variables in
\(L^1\times L^1\). Finally, by the same monotonicity and translation invariance arguments used in the
proof of the general dual representation, the effective dual variables may be restricted
to positive and normalized elements. Thus they correspond to probability measures
absolutely continuous with respect to \(\P\). Consequently, the representation of \(U\)
can be taken over \(\mathcal P\).
\end{proof}

We now show an example that satisfy the conditions on \zcref[S]{theorem dual U}. Specifically, from a convex risk measure $\rho$, we construct an uncertainty set such that the robust risk measure is concave and the consolidated uncertainty set is set-concave.
\begin{example}
 For $\theta>0$, let
$\mathrm{ENT}_{\theta}:L^\infty\to\mathbb R$ be the entropic risk measure defined by
$\mathrm{ENT}_{\theta}(X):=\frac{1}{\theta}\log \E\left[e^{-\theta X}\right]$.

Let $\gamma,\eta>0$ and define $\rho(Y):=\mathrm{ENT}_{\gamma}(Y)$. Then $\rho$ is a convex risk measure. Now define
$R(X):=-\mathrm{ENT}_{\eta}(-X)$. Since $\mathrm{ENT}_{\eta}(-X)=\frac{1}{\eta}\log\E[e^{\eta X}]$, the map $R$ is monetary and concave.

Let $C_m:=\{W\in L^\infty:\|W\|_\infty\le m\}$ for some $m>0$, that is, $C_m$ is a $\lp{}$ closed ball centred around $0$. Moreover, $\sup_{W\in C_m}\rho(W)=m$. Define the convex uncertainty set
\[
u(X):=\mathrm{ENT}_{\eta}(-X)+m+C_m
=
\{\mathrm{ENT}_{\eta}(-X)+m+W:W\in C_m\}.
\]
Then, by translation invariance of $\rho$,
\begin{align*}
R^{u,\rho}(X)
=
\sup_{Y\in u(X)}\rho(Y)
&=
\sup_{W\in C_m}\rho(\mathrm{ENT}_{\eta}(-X)+m+W)
\\
&=
-\mathrm{ENT}_{\eta}(-X)-m+\sup_{W\in C_m}\rho(W)
=
R(X).
\end{align*}
Thus the robust risk measure generated by $(u,\rho)$ is exactly $R$ and the associated consolidated uncertainty set is
\[
U(X)=\{Y\in L^\infty:\rho(Y)\le R(X)\}.
\]
Equivalently, $U(X)=A^\rho-R(X)=A^\rho+\mathrm{ENT}_{\eta}(-X)$. As $R$ is concave, \zcref[S]{theo:equiv_R_U}, \zcref[S]{theo:equiv_R_U_conc} yields that $U$ is set-concave. 
Finally, we compute the dual set of $U$. For $\Q\in\mathcal P$, let
$H(\Q|\P):=\E\left[\frac{d\Q}{d\P}\log\left(\frac{d\Q}{d\P}\right)\right]$ denote the relative entropy, note that for $Q \in \ba \setminus \mathcal{P}, H(\Q|\P) = \infty $. The variational representation of the entropic risk measure gives
$\alpha_\rho(\Q_2)=\frac{1}{\gamma}H(\Q_2|\P)$ and
\[
\beta_R(\Q_1):=\sup_{X\in L^\infty}\left\{\E_{\Q_1}[X]+R(X)\right\}
=
\frac{1}{\eta}H(\Q_1|\P).
\]
By \zcref[S]{cor:exact_equality_delta_g} the support function of the graph of $U$ satisfies
\[
\delta_U(\Q_1,-\Q_2)
=
\beta_R(\Q_1)+\alpha_\rho(\Q_2)
=
\frac{1}{\eta}H(\Q_1|\P)+\frac{1}{\gamma}H(\Q_2|\P).
\]
Consequently, for $\Q_1,\Q_2\in\mathcal P$, the set-valued penalty of $U$ is
\[
\alpha_U(\Q_1,-\Q_2)
=
\left\{
Z\in L^\infty:
\E_{\Q_2}[-Z]\le
\frac{1}{\eta}H(\Q_1|\P)+\frac{1}{\gamma}H(\Q_2|\P)
\right\}.
\]
Thus the dual representation of $U$ is
\[
U(X)
=
\bigcap_{\Q_1,\Q_2\in\mathcal P}
\left\{
Z\in L^\infty:
\E_{\Q_1}[X]+\E_{\Q_2}[-Z]
\le
\frac{1}{\eta}H(\Q_1|\P)+\frac{1}{\gamma}H(\Q_2|\P)
\right\}.
\]
\end{example}

We conclude this section with the same example as in \zcref[S]{sec:dual-risk} but now analysing the dual of the uncertainty set.

\example{
Let \(u(X)=\{Y\in L^\infty\colon \lVert X-Y\rVert\leq \epsilon\}=X+\{Y\in L^\infty\colon \lVert Y\rVert\leq\epsilon\}\), i.e. closed balls in \(L^\infty\). In order to put the example in the setup of \zcref[S]{theorem dual U}, take the base risk measure \(\rho(X)=-\E[X]\). Then \(R(X)=\sup_{\lVert Y-X\rVert\leq\epsilon}-\E[Y]=-\E[X]+\epsilon\), and the consolidated uncertainty set is
\[
U(X)=\{Y\in L^\infty\colon \rho(Y)\leq R(X)\}
=\{Y\in L^\infty\colon \E[Y]\geq \E[X]-\epsilon\}.
\]
We now compute the penalty term for \(\alpha_U(Q_1,-Q_2)\), with \(Q_1,Q_2\in\ba\). By definition, \(\alpha_U(Q_1,-Q_2)=\{Z\in L^\infty\colon \E_{-Q_2}[Z]\leq \delta_U(Q_1,-Q_2)\}\), where
\[
\delta_U(Q_1,-Q_2)
:=
\sup_{X\in L^\infty}\sup_{Y\in U(X)}
\{\E_{Q_1}[X]+\E_{-Q_2}[Y]\}.
\]
If $Q_1\neq Q_2$, it is straightforward to verify that $\delta_U(Q_1,-Q_2)=\infty$. Since \(U(X)=\{Y\in L^\infty\colon \E[Y]\geq \E[X]-\epsilon\}\), we have
\[
\sup_{Y\in U(X)}\E_{-Q_2}[Y]
=
\sup_{\E[Y]\geq \E[X]-\epsilon}-\E_{Q_2}[Y].
\]
If \(Q_2\neq \P\), this supremum is \(+\infty\), since one may add to \(Y\) variables with zero \(\P\)-expectation and nonzero \(Q_2\)-expectation and then rescale. If \(Q_2=\P\), then \(\sup_{Y\in U(X)}\E_{-\P}[Y]=-\E[X]+\epsilon\). Hence
\[
\delta_U(Q_1,-Q_2)
=
\begin{cases}
\epsilon, & Q_1=-Q_2=\P,\\
+\infty, & \text{otherwise}.
\end{cases}
\]
Therefore,
\[
\alpha_U(Q_1,-Q_2)
=
\begin{cases}
H(-\P)-\epsilon, & Q_1=- Q_2=\P,\\
L^\infty, & \text{otherwise}.
\end{cases}
\]
This is consistent with the formula in \zcref[S]{theorem dual U}. Indeed, for \(\rho(X)=-\E[X]\), \(\alpha_\rho(Q_2)=0\) if \(Q_2=\P\), and \(\alpha_\rho(Q_2)=+\infty\) otherwise. Moreover,
\[
\beta_R(Q_1)
=
\sup_{X\in L^\infty}\{\E_{Q_1}[X]+R(X)\}
=
\sup_{X\in L^\infty}\{\E_{Q_1}[X]-\E[X]+\epsilon\},
\]
so \(\beta_R(Q_1)=\epsilon\) if \(Q_1=\P\), and \(\beta_R(Q_1)=+\infty\) otherwise. Thus \(\alpha_U(Q_1,-Q_2)=H(-Q_2)-\alpha_\rho(Q_2)-\beta_R(Q_1)\) reduces exactly to the expression above.
}

\section{Connection between the dual of R and U.}\label{sec:connection}

At this point, a fundamental incompatibility emerges. In Hamel's framework, the dual representation of $U$ hinges on \emph{set-concavity}, whereas in the theory of robust risk measures the dual representation of $R$ relies on the usual convexity of risk measures, which is equivalent to \emph{set-convexity} of the associated consolidated uncertainty set. These two requirements point in opposite directions. Indeed, for a consolidated uncertainty set, set-concavity of $U$ is equivalent, to concavity of $R$. 
Moreover, the set-valued duality of $U$ also requires the set $U(X)$ is convex, which holds if and only if the underlying risk measure $\rho$ is convex.
Thus, this framework naturally combines a convex $\rho$ with a concave $R$. Under two additional assumptions, however, this combination degenerates. First, if both $\rho$ and $R$ are normalized, and second if for each $X$, the uncertainty set around $X$ contains $X$ itself; where the latter implies that $\rho(X)\le R(X)$ for all $X\in L^\infty$. Hence, as $\rho$ is convex and $R$ is concave, normalization forces both functionals to be identical and linear. To see it, we apply  convexity and concavity of $\rho$ and $R$, respectively, together with normalization of both maps, such that for all $X \in \lp{}$ 
\[ 0=2\rho(0) \leq  \rho(X) + \rho(-X)\leq  R(X) + R(-X) \leq 2R(0)=0.\]
Hence, $\rho$ and $R$ are both homogeneous, i.e. $R(\lambda X) = \lambda R(X)$ for all $\lambda \in \R$. Moreover, it holds that $R(X) \geq R(X+Y) + R(-Y) =R(X+Y)  -R(Y)$ and by changing the roles of $X$ and $Y$, $R(X)+R(Y) = R(X+Y)$. The same reasoning applies for $\rho$, hence, both functions are linear and the condition $\rho \leq R$ implies that they are identical. Without normalization of $\rho$ and $R$, the same argument yields affinity rather than linearity, that is the functions satisfy $\rho(X) - \rho(0) = R(X) - R(0)$.

There are however many situations when $X \not\in U(X)$ and thus $\rho$ and $R$ are not affine. Examples include if $X$ is, e.g., a unreliable prior or estimated from poor data. In which case, the uncertainty set $U(X)$ can be constructed to contain the ``true'' (limiting) distribution, however, including $X$ in the uncertainty set would yield an \Uset~too large to be useful. Alternatively, in a financial context, an agent might wish to hedge/replicate $X$. Then $U(X)$ can be viewed as all available financial position to replicate $X$, clearly $X \not\in U(X)$.

Nonetheless, in this section we attempt to establish a connection between both dual representation setups. 
As we can write the penalty function of $U$ as $\alpha_U(Q_1,-Q_2)=\{ Z\in L^\infty\colon \E_{Q_2}[Z]\leq \delta_{U}(Q_1,-Q_2)\}$ we can draw a connection between the penalty function of the \Uset\,to the penalty functions of the robust risk measure $R$ and the risk measure $\rho$. In particular, we can link $\delta_U$, to $\beta_R$,  $ \alpha_\rho$ and $\beta_{g_Q}$, where $\beta_{g_Q}(Q_1):=\sup_{X\in L^\infty}\{\E_{Q_1}[X]+g_Q(X)\}$ plays a similar role as the one of $\beta_R$ for $R$.


\begin{theorem}\label{cor:exact_equality_delta_g}
Let  $u$ be a  translation invariant and order preserving  uncertainty set,  $U$ its  consolidated uncertainty set, $R$ the robust risk measure and $g_Q$ its associated auxiliary function as in \zcref[S]{defi auxiliar}. Then, for all $Q_1,Q_2 \in \ba$ such that $\beta_R(Q_1) + \alpha_\rho(Q_2)<\infty$ we have
\[\delta_U(Q_1, -Q_2)  = \beta_R(Q_1) +\alpha_\rho(Q_2) = \beta_{g^U_{Q_2}}(Q_1).\]


\end{theorem}

\begin{proof}We first show $\delta_U(Q_1, -Q_2)  = \beta_R(Q_1) +\alpha_\rho(Q_2) $.
From the definition of the set-valued conjugate, and noticing that  we have $-Q_2(\Omega) = -1$ for any $Q_2 \in ba_{1,+}$, we can express the penalty $\alpha_U$ via the support function of the graph:
\begin{align}\label{eq for delta = alpha}
\alpha_U(Q_1, -Q_2) = \bigl\{Z \in L^\infty : \E_{-Q_2}[Z] \le \delta_U(Q_1, -Q_2)\bigr\} = H(-Q_2) - \delta_U(Q_1, -Q_2).
\end{align}
On the other hand, from the same steps as those in the prof of \zcref[S]{theorem dual U} regarding the deduction of \eqref{teorem dual U: eq penalty} we get
\begin{align}\label{eq for delta = alpha 2}    
\alpha_U(Q_1, -Q_2) = H(-Q_2) - \alpha_\rho(Q_2) - \beta_R(Q_1),
\end{align}
whenever $ \alpha_\rho(Q_2) $ and $ \beta_R(Q_1)$ are finite.
Equating \eqref{eq for delta = alpha} and \eqref{eq for delta = alpha 2} and isolating the scalar shift yields:
\begin{equation*}\label{eq:delta_U_exact}
\delta_U(Q_1, -Q_2) = \alpha_\rho(Q_2) + \beta_R(Q_1).
\end{equation*}

We now show that $\delta_U(Q_1, -Q_2)   = \beta_{g^U_{Q_2}}(Q_1)$. Denote  $g^U_{Q_2}$ and $\beta_{g^U_{Q_2}}$ as $g$ and $\beta_g$ respectively. 
By the definition of $\delta_U$ and of $g$, direct calculation yields
\begin{align*}
\delta_U(Q_1,-Q_2)
&=
\sup\left\{
\E_{Q_1}[X]+\E_{-Q_2}[Z]:
X\in L^\infty,\ Z\in U(X)
\right\}
\\
&=
\sup_{X\in L^\infty}
\left\{
\E_{Q_1}[X]+
\sup_{Z\in U(X)}\E_{Q_2}[-Z]
\right\}
\\
&=
\sup_{X\in L^\infty}
\left\{
\E_{Q_1}[X]+g(X)
\right\}
=
\beta_g(Q_1).
\end{align*}
\end{proof}

\begin{remark}
 Under the conditions of \zcref[S]{cor:exact_equality_delta_g} and for a coherent risk measure $\rho$, we have  \begin{align*}
    \alpha_U(Q_1,-Q_2)=\{ Z\in L^\infty\colon \E_{Q_2}[-Z]\leq \beta_R(Q_1)\},
\end{align*}
for all $Q_2 \in \ba$ such that $ \alpha_\rho(Q_2)<\infty$. Furthermore, in this case, $\beta_R (Q_1) = \beta_{g^U_{Q_2}} (Q_1)$. That this, the penalty of $U$ can be fully described by the reversed sign convex conjugate of the auxiliary function. 
The second identity in \zcref[S]{cor:exact_equality_delta_g} shows that the support function of the
graph of $U$ is related to the concave conjugate of $g_Q$, not to the usual
risk-measure penalty of $g_Q$. 
Thus, the graph support of $U$ and the epigraph support of $g_Q$ encode dual
information with opposite signs.   
\end{remark}

We end the paper with an example regarding  Bregman-ball uncertainty sets in order to illustrate the
graph-support identity in \zcref[S]{cor:exact_equality_delta_g}. Bregman uncertainty sets have been studied in \cite{Pesenti2024ORL,Pesenti2024WP,tam2026WP} and \cite{Pesenti2026WP}.

\example{
 Let $\phi:I\to\R$ be a proper,
convex, continuously differentiable function on an open interval $I\subseteq\R$,
and let $r\ge 0$. For $X,Z\in L^\infty$ with values in $I$, define
\begin{align*}
B_\phi(Z,X)
:&=
\E\!\left[
\phi(Z)-\phi(X)-\phi'(X)(Z-X)
\right],
\quad \text{and}\\
u(X):&=\left\{Z\in L^\infty:\ B_\phi(Z,X)\le r\right\}.
\end{align*}
For this example, we write \(U=u\). Let $Q\in\mathcal P$ and, under the usual
assumptions to interchange expectation and pointwise optimization, consider
the auxiliary map
\[
g_Q(X):=\sup_{Z\in U(X)}\E_Q[-Z].
\]
Then
\[
g_Q(X)
=
\inf_{\lambda>0}
\left\{
\lambda r+
\E\!\left[
\lambda\left(
\phi^*\left(\phi'(X)-\frac{\frac{dQ}{d\P}}{\lambda}\right)
-\phi^*(\phi'(X))
\right)
\right]
\right\}.
\]
Indeed, this follows from the Lagrangian relaxation of the constraint
$B_\phi(Z,X)\le r$ and the identity
$
\phi^*(\phi'(X))=X\phi'(X)-\phi(X).
$
If $\phi$ is of Legendre type and the infimum is attained at some
$\lambda^*>0$, then the corresponding optimizer is given by
\[
Z^*
=
(\phi')^{-1}
\left(
\phi'(X)-\frac{\frac{dQ}{d\P}}{\lambda^*}
\right).
\]

Consider now the entropic generator
$
\phi(x)=x\log x-x,\ x>0.
$
Then $\phi'(x)=\log x$, $\phi^*(y)=e^y$, and
\[
B_\phi(Z,X)
=
\E\!\left[
Z\log\!\left(\frac{Z}{X}\right)-Z+X
\right],
\qquad X,Z>0.
\]
Hence,
\[
g_Q(X)
=
\inf_{\lambda>0}
\left\{
\lambda r
-
\lambda
\E\!\left[
X\left(1-e^{-\frac{dQ}{d\P}/\lambda}\right)
\right]
\right\}.
\]
If the infimum is attained at $\lambda^*>0$, then
\[
Z^*
=
X e^{-\frac{dQ}{d\P}/\lambda^*}.
\]

We now compute the support scalar that appears in the graph duality of \(U\).
Let \(Q_1\) be a positive finite measure absolutely continuous with respect to
\(\P\), and set
\[
\beta_{g_Q}(Q_1)
:=
\sup_{X\in L^\infty}
\left\{
\E_{Q_1}[X]+g_Q(X)
\right\}.
\]
By \zcref[S]{cor:exact_equality_delta_g}, this scalar satisfies
\[
\delta_U(Q_1,-Q)=\beta_{g_Q}(Q_1).
\]
For the entropic Bregman ball, using the preceding expression for \(g_Q\) and
the same minimax convention as above, we obtain
\[
\beta_{g_Q}(Q_1)
=
\inf\left\{
\lambda r:\ \lambda>0,\ 
\frac{dQ_1}{d\P}
=
\lambda\left(1-e^{-\frac{dQ}{d\P}/\lambda}\right)
\right\},
\]
whenever the constraint is satisfiable, and \(\beta_{g_Q}(Q_1)=+\infty\)
otherwise. In particular, if there exists a unique \(\lambda(Q_1)>0\) such that
\[
\frac{dQ_1}{d\P}
=
\lambda(Q_1)
\left(1-e^{-\frac{dQ}{d\P}/\lambda(Q_1)}\right)
\quad \P\text{-a.s.},
\]
then
\[
\beta_{g_Q}(Q_1)=\lambda(Q_1)r.
\]

Consequently, the set-valued penalty of \(U\) along the direction \((Q_1,-Q)\)
is
\[
\alpha_U(Q_1,-Q)
=
H(-Q)-\beta_{g_Q}(Q_1).
\]
Equivalently, in the unique-\(\lambda\) case,
\[
\alpha_U(Q_1,-Q)
=
H(-Q)-\lambda(Q_1)r.
\]

We emphasize that, in the entropic Bregman case, the density
\[
\lambda\left(1-e^{-\frac{dQ}{d\P}/\lambda}\right)
\]
need not integrate to one. Thus, the preceding display should be read over
positive finite measures \(Q_1\), rather than only over \(\mathcal P\). This is
precisely why the example is better interpreted through the graph-support
identity \(\delta_U(Q_1,-Q)=\beta_{g_Q}(Q_1)\), rather than through the usual
risk-measure penalty \(\alpha_{g_Q}\).
}

\section*{Acknowledgements.}
M. Righi is grateful for the support of Brazilian Research Council (CNPq grants 302614/2021-4 and 401720/2023-3). S. Pesenti acknowledges support from the Natural Sciences and Engineering Research Council of Canada (RGPIN-2025-05847). M. Moresco  is grateful for the support of the Fundação de Amparo à Pesquisa do Estado do Rio Grande do Sul (FAPERGS) under Grant 25/2551-0000969-3

\bibliographystyle{spmpsci}
\bibliography{references} 
\end{document}